\documentclass[12pt]{article}
\usepackage{amsmath}
\usepackage{amssymb,amscd}
\usepackage{bm}
\usepackage{wrapfig}

\usepackage{theorem}

\newtheorem{theorem}{Theorem}[section]
 \newtheorem{definition}[theorem]{Definition}
 \newtheorem{proposition}[theorem]{Proposition}
 \newtheorem{corollary}[theorem]{Corollary}

{\theorembodyfont{\normalfont}
 \newtheorem{example}{Example}[section]
 \newtheorem{remark}{Remark}[section]
 }

\def\thefootnote{\fnsymbol{footnote}}

\newlength{\minitwocolumn}
\catcode`\@=11
\long\def\@makefntext#1{
\protect\noindent \hbox to 3.2pt {\hskip-.9pt
$^{{\eightrm\@thefnmark}}$\hfil}#1\hfill}               

\def\thefootnote{\fnsymbol{footnote}}
\def\@makefnmark{\hbox to 0pt{$^{\@thefnmark}$\hss}}    

\def\ps@myheadings{\let\@mkboth\@gobbletwo
\def\@oddhead{\hbox{}
\rightmark\hfil\eightrm\thepage}
\def\@oddfoot{}\def\@evenhead{\eightrm\thepage\hfil
\leftmark\hbox{}}\def\@evenfoot{}
\def\sectionmark##1{}\def\subsectionmark##1{}}

\font\eightrm=cmr8

\newcommand{\qed}{\nobreak \ifvmode \relax \else
      \ifdim\lastskip<1.5em \hskip-\lastskip
      \hskip1.5em plus0em minus0.5em \fi \nobreak
      \vrule height0.75em width0.5em depth0.25em\fi}

\newcommand{\sbv}[2]{{\{{{#1},{#2}}\}}}
\newcommand{\bsbv}[2]{{\{{{#1},{#2}}\}}}
\newcommand{\ssbv}[2]{{[{{#1},{#2}}]}}
\newcommand{\dsbv}[2]{{[{{#1},{#2}}]_{D}}}

\newcommand{\bracket}[2]{\langle #1\,,#2\rangle}

\newcommand{\floor}[1]{{\lfloor #1 \rfloor}}

\def\bx{\mbox{\boldmath $x$}}
\def\bxi{\mbox{\boldmath $\xi$}}
\def\bq{\mbox{\boldmath $q$}}
\def\bp{\mbox{\boldmath $p$}}

\def\bbx{\mbox{\boldmath $x$}}
\def\bbxi{\mbox{\boldmath $\xi$}}

\def\bomega{\mbox{\boldmath $\omega$}}

\def\tiS{{S_{\calM}}}

\def\balpha{\mbox{\boldmath $\alpha$}}

\newcommand{\bZ}{\mathbb{Z}}
\newcommand{\bR}{\mathbb{R}}

\def\bPhi{\mbox{\boldmath $\Phi$}}

\newcommand{\calA}{{\cal A}}
\newcommand{\calB}{{\cal B}}
\newcommand{\calC}{{\cal C}}

\newcommand{\calI}{{\cal I}}
\newcommand{\calJ}{{\cal J}}

\newcommand{\calL}{{\cal L}}
\newcommand{\calM}{{\cal M}}
\newcommand{\calN}{{\cal N}}
\newcommand{\calO}{{\cal O}}

\newcommand{\calX}{{\cal X}}

\newcommand{\Q}{{\kern.24em\vrule width.04em height1.4ex%
                 depth-.05ex\kern-.26em\mathsf Q}}
\newcommand{\C}{{\kern.24em\vrule width.04em height1.4ex%
                 depth-.05ex\kern-.26em\mathsf C}}

\newcommand{\Map}{{\rm Map}}
\newcommand{\ev}{{\rm ev}}

\newcommand{\nom}{\nonumber}
\newcommand{\beq}{\begin{equation}}
\newcommand{\eeq}{\end{equation}}
\newcommand{\bea}{\begin{eqnarray*}}
\newcommand{\eea}{\end{eqnarray*}}
\newcommand{\beqa}{\begin{eqnarray}}
\newcommand{\eeqa}{\end{eqnarray}}

\newcommand{\bJ}{\boldsymbol{J}}

\newcommand{\veta}{\boldsymbol{\eta}}
\newcommand{\vxi}{\boldsymbol{\xi}}
\newcommand{\vchi}{{\boldsymbol{\chi}}}

\newcommand{\del}{\partial}

\newcommand{\courant}[2]{{[{{#1},{#2}}]_D}}

\newcommand\gh{{\rm gh}}

\newcommand{\cJ}{{J_{cl}}}
\newcommand{\cepsilon}{{\epsilon_{cl}}}

\newcommand{\proj}{{pr}}

\newcommand{\omegab}{{\omega_{T^*[n+1]\calM}}}
\newcommand{\omegas}{{\omega_{\calM}}}
\def\bomegab{{\boldsymbol{\omega}_{T^*[n+1]\calM}}}
\def\bomegas{{\boldsymbol{\omega}_{\calM}}}

\newcommand{\bvartheta}{{\boldsymbol{\vartheta}}}

\newcommand{\varthetas}{{\vartheta_{\calM}}}
\newcommand{\bvarthetas}{{\boldsymbol{\vartheta}_{\calM}}}

\def\bTheta{\mbox{$S$}}

\newcommand{\rd}{{\mathrm{d}}}

\newcommand{\Cn}{{C_{\leq n}}}
\newcommand{\Bn}{{B_{\leq n}}}

\def\bbd{{D}}

\newcommand{\CAn}{{\calC\calA_{\leq n}}}
\newcommand{\CAone}{{\calC\calA_{\leq 1}}}

\begin{document}


\baselineskip 0.7cm

\begin{titlepage}
\begin{flushright}
\end{flushright}

\vskip 1.35cm
\begin{center}
{\Large \bf
Current algebras from QP-manifolds in general dimensions
}
\vskip 1.2cm
Noriaki Ikeda$^{1
}$%
\footnote{E-mail:\
nikeda@se.ritsumei.ac.jp
}
and
Xiaomeng Xu${}^2$%
\footnote{E-mail:\
xxu@bicmr.pku.edu.cn
}
\vskip 0.4cm
{\it
$^1$
Department of Mathematical Sciences, Ritsumeikan University \\
Kusatsu 525-8577, Japan
\\
$^2$
Scool of Mathematical Sciences and Beijing International Center for Mathematical Research, Peking University, Beijing 100871, China 
}
\vskip 0.4cm

\today

\vskip 1.5cm

\begin{abstract}
We propose a new unified formulation of the current algebra theory in general dimensions in terms of supergeometry. We take a QP-manifold, i.e. a differential graded (dg) symplectic manifold, as a fundamental framework. A Poisson bracket in a current algebra is constructed by the so called derived bracket of the graded Poisson structure induced from the above QP-structure. By taking a canonical transformation on a QP-manifold, correct anomalous terms in physical theories are derived. A large class of current algebras with and without anomalous terms (central extensions) are constructed from the above structure. Moreover, using this formulation, a new class of current algebras related higher structures are systematically obtained.
\end{abstract}
\end{center}
\end{titlepage}

\renewcommand{\thefootnote}{\alph{footnote}}

\setcounter{page}{2}


\rm

\section{Introduction}
\noindent
A current algebra is a Poisson algebra for conserved currents 
in a quantum field theory \cite{Treiman}. 
For instance, 
the commutation relation of currents $J$'s for 
a nonabelian Lie group are
\begin{eqnarray}
[J_a(\sigma),J_b(\sigma^{\prime})]
= {f_{ab}}^c J_c(\sigma) \delta^n(\sigma-\sigma^{\prime}),
\end{eqnarray}
where ${f_{ab}}^c$ is the structure constant of the corresponding Lie algebra
and $\sigma$ and $\sigma^{\prime}$ are local coordinates on a space.
In general, 
extra terms called anomalous terms
appear in the commutators:
\begin{eqnarray}
[J_a(\sigma),J_b(\sigma^{\prime})]
= {f_{ab}}^c J_c(\sigma) \delta(\sigma-\sigma^{\prime})
+ \mbox{ anomalous terms},
\end{eqnarray}
which contain derivatives of the delta function
$\delta(\sigma-\sigma^{\prime})$.
They appear in many important physical models such as
a chiral current in the gauge theories,
the Kac-Moody algebra in the theories in $(1+1)$ dimensions, etc.
These breaking terms of Poisson brackets are important 
to analyze a consistent quantum field theory
\cite{Faddeev:1984jp, Faddeev:1985iz}.

Recently, current algebras with higher structures have been constructed.
This has been firstly observed in \cite{Alekseev:2004np}, 
in which a Courant algebroid structures 
appear in a current algebra of a $2$-dimensional sigma model.
Current algebras with higher algebroid structures have also been 
formulated, which include current algebras of $p$-brane sigma models 
\cite{Bonelli:2005ti, Guttenberg, Ikeda:2011ax}. 
However, we do not have a general geometrical framework
to describe anomalous terms.
It is important to construct a general formalism in order 
to understand algebraic and geometric structure of a generalized current algebra theory.
This article proposes a new geometric formulation based on supergeometry 
to understand higher algebroid current algebras and their anomaly terms.

Supergeometry appears
in many subjects in mathematics and physics.
It is closely connected to the physical formalisms
such as supersymmetry \cite{Manin}, 
and the BRST-BV-BFV formalism \cite{Batalin:jr, Batalin:1983pz, Henneaux:1992ig}.
A recent interesting application of supergeometry to physical theories is 
the Alexandrov-Kontsevich-Schwarz-Zaboronsky (AKSZ) formalism of topological sigma models
\cite{Alexandrov:1995kv, Cattaneo:2001ys, Roytenberg:2006qz, Ikeda:2012pv}.
Moreover structures from various algebroids 
such as Poisson structures, Lie algebroids, Lie $n$-algebroids
and Dirac structures are naturally formulated by supergeometry
\cite{Roy01, Severa:2001}.
Our theory is directly connected to above structures.

Recently the authors have introduced 
notion of 
a tower of (twisted)
differential graded symplectic structures satisfying compatibility conditions 
on two graded manifolds $\calM$ and $T^*[n+1]\calM$ \cite{Ikeda:2013wh}.
They unify geometric structures nontwisted and twisted differential graded symplectic manifolds (which can break the classical master equation).
Moreover they have provided mathematical structure of 
AKSZ sigma models with boundary.
In this paper, we show similar structures are also 
proper framework for the current algebra theory.

In this paper, we construct a general supergeometric formalism of 
current algebra theory. 
The anomalous current algebras are derived from 
twist of an underlying graded symplectic structure.
We work with two mapping spaces
$\Map(T[1]\Sigma_n,\calM)$ and
$\Map(T[1]\Sigma_n,T^*[n+1]\calM)$,
where $\Sigma_n$ is a compact oriented manifold of dimension $n$
and $\calM$ is a graded symplectic manifold of degree $n$.
$\calM$ corresponds to an extended target space of the AKSZ formalism 
including physical fields, ghosts and antifields.
$(T^*[n+1]\calM, \calM)$ are a graded manifold and its cotangent bundle pair
with some consistent structure.
In this setting, currents on
$\Map(T[1]\Sigma_n,\calM)$
with anomalous terms are derived from a Poisson algebra of functions
on the extended mapping space $\Map(T[1]\Sigma_n,T^*[n+1]\calM)$.
We introduce a canonical transformation called twisted pullback,
and a pullback derived bracket.
The formalism in this paper may be called as
''the AKSZ-BFV formalism of current algebra theories''.

A large class of known current algebras are formulated as 
our supergeometric current algebra theory.
They contain not only the current algebras 
of Lie algebras and Kac-Moody algebras,
but also current algebras in AKSZ sigma models.
Moreover we find new current algebras with structure of higher algebroids
including all the generalized current algebras analyzed in the recent papers
\cite{Alekseev:2004np, Bonelli:2005ti, Guttenberg, Ikeda:2011ax}.

Recently, similar but slightly different construction of current algebras based on Q-manifolds was given in \cite{Arvanitakis:2021wkt}.

The paper is organized as follows.
In section 2, a QP-manifold, twisting
are prepared as preliminary.
In section 3, we analyze properties of the derived brackets between functions
and target space structures.
In section 4, we propose a supergeometric formalism of 
current algebras and discuss its relation with classical current algebras. 
We also give many examples which illuminate our theory.
In section 5, current algebras in the AKSZ sigma models
are reconstructed as special cases of our theory.
Section 6 is devoted to future outlook.
In the appendix, we consider an example in a target space.

\section{Preliminary: QP-manifolds and twist
}
\noindent
We explain mathematical tools in our articles,
a QP-manifold, twist and a Lagrangian submanifold.
For details on the background and definitions, we refer to \cite{Severa:2001, Voronov:2001qf, Roytenberg:2006qz, Cattaneo:2010re, Ikeda:2012pv, Cueca:2019rmv}, etc.

First, we consider our theory on a graded manifold, a manifold with $\bZ_{\geq 0}$-grading. The precise definition is as follows.
\begin{definition}
A graded manifold $\calM$ on a smooth manifold $M$
is defined as a ringed space with a structure sheaf $\calO$
of a nonnegatively graded commutative
algebra over an ordinary smooth manifold $M$.
$\calM$ is locally isomorphic to
$C^{\infty}(U) \otimes S^{\bullet}(V)$,
where
$U$ is a local chart on $M$,
$V$ is a graded vector space and $S^{\bullet}(V)$ is a free
$\bZ$-graded commutative ring on $V$.
\end{definition}
Grading is called \textsl{degree} and denoted by $|x|$ for an element $x$.
Grading is compatible with supermanifold grading,
that is, a variable of even degree is commutative and
one of odd degree is anticommutative, i.e.,  for elements with definite degree
 $x$ and $y$, the product satisfies $xy = (-1)^{|x||y|} yx$.

%
%
The structure sheaf $\calO$ is called a space of functions on $\calM$ and 
denoted by $C^{\infty}(\calM)$.
In this paper,  a graded manifold is always of nonnegative degree.
Such a graded manifold is called an {\it N-manifold}.

\begin{example}\label{ex01}
Let $M$ be a smooth manifold and take a graded manifold $\calM = T[1]M$.
Take local coordinates $(x^i, q^i)$ of degree $(0, 1)$ on $T[1]M$,
where $x^i$ is a local coordinate on $M$ and $q^i$ is an odd 
local coordinate on the fiber $T[1]$.
A function $f(x, q)$ on $T[1]M$ is locally 
\begin{eqnarray}
f(x, q) &=& \sum_{k=0}^d \frac{1}{k!} 
f^{(k)}_{i_1, \ldots, i_k}(x) q^{i_1} \ldots q^{i_k}.
\end{eqnarray}
$f(x, q)$ is equivalent to the differential form,
\begin{eqnarray}
f(x, q) &=& \sum_{k=0}^d \frac{1}{k!} 
f^{(k)}_{i_1, \ldots, i_k}(x) \rd x^{i_1} \wedge \ldots \wedge \rd x^{i_k}.
\end{eqnarray}
Thus, the space of function on $T[1]M$ is equivalent to the space of differential forms, $C^{\infty}(T[1]M) \simeq \Omega^{\bullet}(M)$.
\end{example}

\begin{example}\label{ex02}
Let $M$ be a smooth manifold and take a graded manifold $\calM = T^*[1]M$.
Take local coordinates $(x^i, p_i)$ of degree $(0, 1)$ on $T^*[1]M$,
where $x^i$ is a local coordinate on $M$ and $p_i$ is an odd 
local coordinate on $T^*[1]$.
A function $f(x, p)$ on $T^*[1]M$ is locally 
\begin{eqnarray}
f(x, q) &=& \sum_{k=0}^d \frac{1}{k!} 
f^{(k) i_1, \ldots, i_k}(x) p_{i_1} \ldots p_{i_k}.
\end{eqnarray}
$f(x, p)$ is equivalent to the multivector field,
\begin{eqnarray}
f(x, p) &=& \sum_{k=0}^d \frac{1}{k!} 
f^{(k) i_1, \ldots, i_k}(x) \frac{\partial}{\partial x^{i_1}} \wedge \ldots \wedge  \frac{\partial}{\partial x^{i_k}}.
\end{eqnarray}
Thus, the space of function on $T^*[1]M$ is equivalent to the space of multivector fields, $C^{\infty}(T^*[1]M) \simeq \mathfrak{X}^{\bullet}(M)$.
\end{example}

\begin{example}\label{ex03}
A graded manifold is not necessarily equivalent to 
a space of sections on a vector bundle.
For a rank $d$ vector bundle $E$ over $M$, consider 
the graded cotangent bundle $\calM=T^*[2]E[1]$.
Take local coordinates $(x^i, q^a)$ of degree $(0, 1)$ on $E[1]$,
where $x^i$ is a local coordinate on $M$ and $q^a$ is a
local coordinate of degree one on the fiber.
Let $(\xi_i, p_a)$ be conjugate coordinates of degree $(2, 1)$ on $T^*[2]$.
In general, degree one coordinates $q^a$ and $p_a$, $(a = 1, \cdots, d)$, are combined to 
$\eta^a$, $(a = 1, \cdots, 2d)$,
by introducing a fiber metric $k = \bracket{-}{-}$, 
where $p_a = k_{ab} q^a$.
Using local coordinates $(x^i, \eta^a, \xi_i)$, 
local coordinate transformations of $\calM=T^*[2]E[1]$
\begin{eqnarray}
x^{\prime i} &=& x^{\prime i}(x),
\label{lctransf21}
\\
\eta^{\prime a} &=& M^a_b(x) \eta^{b},
\label{lctransf22}
\\
\xi^\prime_i &=& \frac{\partial x^j}{\partial x^{\prime i}} \xi_j
- \frac{1}{2} k_{ac} M^a_b(x) \frac{\partial M^b_d}{\partial x^i}(x) 
\eta^{c} \eta^{d},
\label{lctransf23}
\end{eqnarray}
where $x^{\prime i}(x)$ is a change of the local coordinate on $M$
and $M^a_b(x)$ is a transition function of the fiber of $E$.
The second term of the right hand side of Eq.~\eqref{lctransf23}
is not realized by any fiber bundle over $M$.
\end{example}

A graded vector field on a graded manifold is 
a derivation on $C^{\infty}(\calM)$
and a graded differential form on a graded manifold is 
an element on the dual space of the space of graded vector fields. 
A graded de Rham differential on a graded manifold $\calM$ 
is denoted by $\delta$. \footnote{$\rd$ is denoted by the normal 
de Rham differential on the normal manifold $M$.}

An N-manifold $\calM$ equipped with a graded symplectic structure
$\omega$ of degree $n$, denoted by $({\calM},\omega)$, 
is called a P-manifold of degree $n$.
$\omega$ is also called a P-structure.
The graded Poisson bracket on $C^\infty ({\cal M})$ is defined
from the graded symplectic structure $\omega$ on ${\cal M}$ as
$\sbv{f}{g} = (-1)^{|f|+n+1} \iota_{X_f} \iota_{X_g}\omega$, 
similar to normal Poisson brackets.
Here a Hamiltonian vector field $X_f$ is defined by the equation
$\iota_{X_f}\omega= - \delta f$ for $f\in C^{\infty}({\calM})$
The graded Poisson bracket of degree $n$
are graded skewsymmetric and satisfies the graded Leibniz 
rule and the graded Jacobi identity. Concrete formulas are
\begin{eqnarray}
\sbv{f}{g}&=&-(-1)^{(|f| + n)(|g| + n)} \sbv{g}{f},
\label{Poissonidentity1}
\\
\sbv{f}{g h}&=&\sbv{f}{g} h
+ (-1)^{(|f| + n) |g| } g \sbv{f}{h},
\label{Poissonidentity2}
\\
\sbv{fg}{h}&=&f \sbv{g}{h}
+ (-1)^{|g| (|h|  + n)} \sbv{f}{h} g,
\label{Poissonidentity3}
\\
\sbv{f}{\sbv{g}{h}}&=&\sbv{\sbv{f}{g}}{h}
+ (-1)^{(|f|+ n)(|g| + n)} \sbv{g}{\sbv{f}{h}}.
\label{Poissonidentity4}
\end{eqnarray}

A graded manifold $\calM$ is called a Q-manifold if there exists 
a vector field $Q$ of degree $+1$ satisfying $Q^2=0$.
$Q$ is called a homological vector field, or a Q-structure.
If $\omega$ and $Q$ are compatible, it is called a \textit{QP-manifold}.
\begin{definition}
A triple $(\calM, \omega, Q)$ is called a {\rm QP-manifold} of degree $n$
and its structure is called a {\rm QP structure},
if 
$L_Q \omega =0$ {\rm \cite{Schwarz:1992nx}}.\footnote{A QP-manifold is also called a \textsl{symplectic NQ-manifold, or a differential graded (dg) symplectic manifold.}}
\end{definition}
In a QP-manifold of degree nonzero, there always exists a Hamiltonian function 
$\Theta\in C^{\infty}(\calM)$
of $Q$ 
satisfying
\beq
Q=\sbv{\Theta}{-}.
\eeq
The function $\Theta$ if of degree $n+1$. 
The homological condition, $Q^2=0$, implies that
$\Theta$ is a solution of the \textsl{classical master equation},
\begin{equation}
\sbv{\Theta}{\Theta}=0.
\label{cmaseq}
\end{equation}
A function $\Theta$ satisfying (\ref{cmaseq}) is also called a homological function or a Hamiltonian function. 

\begin{example}\label{ex04}
Take a Lie algebra $\mathfrak{g}$ and consider 
the exterior algebra $\wedge^{\bullet} \mathfrak{g}$ over $\mathfrak{g}$.
$\wedge^{\bullet} (\mathfrak{g} \oplus \mathfrak{g}^*)$ is equivalent to
$C^{\infty}(T^*[1]\mathfrak{g}[1])$.
Take a local coordinate $c^a$ of degree one on $\mathfrak{g}[1]$
and a local coordinate $b_a$ of degree one $T^*[1] \simeq \mathfrak{g}^*[1]$.
Then, $\omega = \delta c^a \wedge \delta b_a$ is the canonical graded symplectic form of degree two.
The induced Poisson bracket satisfies $\sbv{c^a}{b_b}= \delta^a_b$.
Since $\omega$ is of degree two, a homological function $\Theta$ must be of degree three. 
We define $\Theta$ as
\begin{eqnarray}
\Theta &=& \frac{1}{2} C_{ab}^c c^a c^b b_c,
\end{eqnarray}
with a constant $C_{ab}^c$.
Then, the condition \eqref{cmaseq} is equivalent that 
$\mathfrak{g}$ is a Lie algebra.
In fact, Eq.~\eqref{cmaseq} is satisfied if and only if 
$C_{ab}^c$ is the structure constant of Lie algebra $\mathfrak{g}$.

Then $Q = \sbv{\Theta}{-}$ gives the Chevalley-Eilenberg differential.
\end{example}

\begin{example}\label{ex05}
Take a graded manifold $T^*[1]M$ in Example \ref{ex02}.
The canonical graded symplectic form is 
$\omega = \delta x^i \wedge \delta p_i$ of degree one.
The induced Poisson bracket is $\sbv{x^i}{p_j} = \delta^i_j$.
The most general homological function is of degree two and 
$\Theta = \frac{1}{2} \pi^{ij}(x) p_i p_j$.
where $\pi^{ij}(x) = - \pi^{ji}(x)$.
Eq.~\eqref{cmaseq} gives a condition in $\pi^{ij}$, which is equivalent that 
the bivector field 
$\pi = \frac{1}{2} \pi^{ij} \frac{\partial}{\partial x^i} \wedge  \frac{\partial}{\partial x^j}$ is a Poisson bivector field.
Then,
\begin{eqnarray}
\{f(x), g(x) \}_{PB} &=& \frac{1}{2} \pi^{ij}(x) \frac{\partial f}{\partial x^i}\frac{\partial g}{\partial x^j}
\end{eqnarray}
is a normal Poisson bracket on $M$.
\end{example}

\begin{example}\label{ex052}
Let $E$ be a vector bundle over a smooth manifold $M$.
Consider the graded cotangent bundle $T^*[2]E[1]$ discussed 
in Example \ref{ex03}. 
A QP-manifold of degree $2$ is constructed in $T^*[2]E[1]$.
$T^*[2]E[1]$ has the following canonical graded symplectic form of degree $2$ since it is a cotangent bundle,
\begin{eqnarray}
\omega = \delta x^i \wedge \delta \xi_i 
+ \frac{1}{2} \delta \eta^a \wedge \delta(k_{ab} \eta^b).
\end{eqnarray}
A homological function is of degree $3$
and the most form is 
\begin{eqnarray}
\Theta = \rho{}^i_{a} (x) \xi_{i} \eta{}^a 
+ \frac{1}{3!} H_{abc} (x) \eta{}^a \eta{}^b \eta{}^c,
\label{courantQ}
\end{eqnarray}
where $\rho{}^i_{a} (x)$ and $H_{abc} (x)$ are
local functions of $x$.

The Q-structure condition $\sbv{\Theta}{\Theta}=0$
imposes the following relations on these functions:
\begin{eqnarray}
&& k^{ab} \rho{}^i_{a} \rho{}^j_{b} = 0, \nonumber \\ 
&& \frac{\partial \rho{}^i_{b}}{\partial x^j} \rho{}^j_{c}
- \frac{\partial \rho{}^i_{c}}{\partial x^j} \rho{}^j_{b}
+ k^{ef}  \rho{}^i{}_{e}  H_{fbc} = 0, \nonumber \\
&& \left( \rho{}^i{}_{d} \frac{\partial H_{abc}}{\partial x^i}
- \rho{}^i{}_{c} \frac{\partial H_{dab}}{\partial x^i}
+ \rho{}^i{}_{b} \frac{\partial H_{cda}}{\partial x^i}
- \rho{}^i{}_{a} \frac{\partial H_{bcd}}{\partial x^i} 
\right) 
\nonumber \\
&& \qquad
+ k^{ef} (H_{eab}  H_{cdf} 
+ H_{eac} H_{dbf} 
+ H_{ead} H_{bcf})
= 0.
\label{courantrelation}
\end{eqnarray}
We can prove that 
these identities (\ref{courantrelation}) are the same as the local coordinate 
expressions of the Courant algebroid conditions
on a vector bundle $E$. We explain it in Example \ref{ex06} in the next section.
\end{example}

\begin{definition}
Let $(\calM, \omega, \Theta)$
be a QP-manifold of degree $n$.
For any function $f  \in C^{\infty}(\calM)$ and a function 
$\alpha \in C^{\infty}(\calM)$ of degree $n$,
the operation 
$e^{\mathrm{ad} (\alpha)} f = f +
\sbv{f}{\alpha}
+ \frac{1}{2} \sbv{\sbv{f}{\alpha}}{\alpha}
+ \frac{1}{3!} \sbv{\sbv{\sbv{f}{\alpha}}{\alpha}}{\alpha}
+ \cdots$ is called a 
{\rm twist} by $\alpha$.
\end{definition}
A twist by $\alpha$ satisfies
$\sbv{e^{\mathrm{ad} (\alpha)} f}{e^{\mathrm{ad} (\alpha)} g}
= e^{\mathrm{ad} (\alpha)} \sbv{f}{g}$ for every $f, g  \in C^{\infty}(\calM)$
since $\sbv{-}{-}$ satisfies the Jacobi identity.
Especially, if $\sbv{\Theta}{\Theta}=0$, 
$\sbv{e^{\mathrm{ad} (\alpha)} f}{e^{\mathrm{ad} (\alpha)} g}=0$.
Thus, the following proposition is obtained.
\begin{proposition}
$(\calM, \omega, e^{\mathrm{ad} (\alpha)} \Theta)$
is a QP-manifold if $(\calM, \omega, \Theta)$ is a QP-manifold.
\end{proposition}

\section{Derived brackets, Lagrangian submanifolds and horizontal homological functions
}
\subsection{Derived brackets}
In order to apply a QP-manifold to the current algebra theory, we consider
a graded cotangent bundle $T^*[n+1]\calM$ of a graded manifold $\calM$.
We denote the natural projection by $\proj: T^*[n+1]\calM \rightarrow \calM$
and the inclusion as zero sections by $\iota: \calM \rightarrow T^*[n+1]\calM$.

Assume a QP-manifold structure $(T^*[n+1]\calM, \omegab, Q)$ 
on the graded cotangent bundle 
with a canonical symplectic form $\omegab = \omega_{can}$
of degree $n+1$ and a homological vector field $Q$.
Note that $\calM$ is trivially a Lagrangian submanifold of $T^*[n+1]\calM$.
We take a homological function $\Theta$ such that  $Q(-) = \bsbv{\Theta}{-}$.

In this section, before considering a current algebra,
we analyze properties of functions on $T^*[n+1]\calM$.
It is regraded as the current algebra theory 
in zero dimension (in the mechanics).
In the next section, we consider functions 
on the mapping space, i.e., current algebras.

We construct a new bilinear bracket called a 
\textit{derived bracket} \cite{Kosmann-Schwarzbach:2007} from 
the graded Poisson bracket $\bsbv{-}{-}$ as
\begin{eqnarray}
\dsbv{f}{g} := - \bsbv{\bsbv{f}{\Theta}}{g}.
\end{eqnarray}
for 
$f, g \in C^{\infty}(T^*[n+1]\calM)$. 
The derived bracket does not necessarily satisfy the identities 
(\ref{Poissonidentity1})--(\ref{Poissonidentity3}) of a graded Poisson bracket.
However, the derived bracket satisfies the the Jacobi type identity (\ref{Poissonidentity4}) from $\sbv{\Theta}{\Theta}=0$ \footnote{It is also called the Leibniz identity.}.
In general, other identities do not hold. 
Especially, it is not necessarily graded skewsymmetric.

\begin{example}\label{ex06}
We consider the QP-manifold of degree $2$, $T^*[2]E[1]$ in Example \ref{ex052}.

The space of functions $C^{\infty}(T^*[2]E[1])$ is expanded by degree, 
$C^{\infty}(T^*[2]E[1]) = \oplus_{k \geq 0} C_k(T^*[2]E[1])$,
where $C_k(T^*[2]E[1])$ is the space of degree $k$ functions.
A QP-manifold of degree two induces a Courant algebroid \cite{Roy01}.
For it, we consider the degree zero and degree one subspace $C_0(T^*[2]E[1]) \oplus C_1(T^*[2]E[1])$. $C_0(T^*[2]E[1])$ is equivalent to the space of functions on $M$, $C^{\infty}(M) \simeq C_0(T^*[2]E[1])$.
$C_1(T^*[2]E[1])$ is equivalent to the space of sections of $E$, $\Gamma(E) \simeq C_1(T^*[2]E[1])$.

The graded Poisson bracket and \textit{derived brackets} gives operations of induced algebroid. Nontrivial operations are
\begin{eqnarray}
\bracket{a}{b} &=& \bsbv{a}{b},
\label{operation1}
\\
\rho(a) f &=& \bsbv{\bsbv{a}{\Theta}}{f},
\label{operation2}
\\
\courant{a}{b} &=& \bsbv{\bsbv{a}{\Theta}}{b}.
\label{operation3}
\end{eqnarray}
The homological condition $\bsbv{\Theta}{\Theta}=0$ gives identities
between three operations \eqref{operation1}--\eqref{operation3}.
These conditions are equivalent to the following definition of a Courant algebroid \cite{Courant, LWX}.
\begin{definition}\label{courantdefinition}
A Courant algebroid is a vector bundle $E \longrightarrow M$,
and it has a nondegenerate symmetric bilinear form
$\bracket{-}{-}$ 
on the bundle, a bilinear operation $\courant{-}{-}$ on $\Gamma (E)$,
and a bundle map called an anchor map,
$\rho: E \longrightarrow TM$, satisfying the following properties:
%
\begin{eqnarray}
&& 1, \quad \courant{e_1}{\courant{e_2}{e_3}} = \courant{\courant{e_1}{e_2}}{e_3} + \courant{e_2}{\courant{e_1}{e_3}}, 
  \label{courantdef1}
\\
&& 2, \quad \rho(\courant{e_1}{e_2}) = [\rho(e_1), \rho(e_2)], 
  \label{courantdef2}
\\
&& 3, \quad \courant{e_1}{f e_2} = f (\courant{e_1}{e_2})
+ (\rho(e_1)f)e_2, 
  \label{courantdef3}
 \\
&& 4, \quad \courant{e_1}{e_2} = \frac{1}{2} {\cal D} \bracket{e_1}{e_2},
  \label{courantdef4}
\\ 
&& 5, \quad \rho(e_1) \bracket{e_2}{e_3}
= \bracket{\courant{e_1}{e_2}}{e_3} + \bracket{e_2}{\courant{e_1}{e_3}},
  \label{courantdef5}
\end{eqnarray}
where 
$e_1, e_2$, and $e_3$ are sections of $E$, $f$ is a function on
$M$ and 
${\cal D}$ is a map from the space of functions on $M$ to $\Gamma (E)$, 
defined as 
$\bracket{{\cal D}f}{e} = \rho(e) f = \bsbv{\bsbv{e}{\Theta}}{f}$.
\end{definition}
\end{example}

\subsection{Subspaces closed under derived brackets}
\noindent
%
For current algebras,
we define a subspace $\Cn(T^*[n+1]\calM)$ of $C^{\infty}(T^*[n+1]\calM)$.
\begin{definition}
A space of functions of degree equal to or less than $n$ on $T^*[n+1]\calM$ is
denoted by $\Cn(T^*[n+1]\calM)$.
i.e., 
it is defined by $\Cn(T^*[n+1]\calM)
= \{f \in C^{\infty}(T^*[n+1]\calM)| |f| \leq n \}$.
\end{definition}
\begin{proposition}
The subspace $\Cn(T^*[n+1]\calM)$ is closed not only under the
graded Poisson bracket $\bsbv{-}{-}$, but also
under the derived bracket $\courant{-}{-} = \bsbv{\bsbv{-}{\Theta}}{-}$.
i.e., for every functions $f, g \in \Cn(T^*[n+1]\calM)$,
\begin{eqnarray}
&& \bsbv{f}{g} \in \Cn(T^*[n+1]\calM),
\\
&& \bsbv{\bsbv{f}{\Theta}}{g} \in \Cn(T^*[n+1]\calM).
\end{eqnarray}
\end{proposition}
It is proved by simple degree counting.

A graded submanifold $\calL \subset \calN$ of $\calN$ is a pair 
$\calL = (L, \calO^{\prime})$ of 
a submanifold $L \subset N$ and 
a subshief $\calO^{\prime} \subset \calO$.
Here $\calO^{\prime}$ is locally a subalgebra of
$C^{\infty}(U) \otimes S^{\bullet}(V)$.

If for every $f, g \in C^{\infty}(\calL)$ 
$\bsbv{f}{g} = 0$, $\calL$ is called \textit{isotropic}.
The condition is also denoted by $\omega|_{\calL} =0$.
If a isotropic submanifold $\calL$ satisfies $\mathrm{dim}(\calL) = \frac{1}{2} \mathrm{dim}(\calN)$, it is called a Lagrangian submanifold. 
A space of functions on $\calL$ is denoted by $C^{\infty}(\calL)$.
The inclusion map is denoted by 
$\iota: \calL \rightarrow \calN$, 

For a Lagrangian submanifold $\calL \subset \calN = T^*[n+1]\calM$,
let $C^{\infty}(\calL)$ be the space of functions.
Using the map $\iota: \calL \rightarrow \calN$, $C^{\infty}(\calL)$
is regarded as a subspace of $C^{\infty}(T^*[n+1]\calM)$.
Under this identification, for every $f, g \in C^{\infty}(\calL)$,
$\sbv{f}{g}=0$.
Thus the derived bracket $\bsbv{\bsbv{-}{-}}{-}$ satisfies identities
(\ref{Poissonidentity1})--(\ref{Poissonidentity4}).
$\bsbv{\bsbv{-}{-}}{-}$ gives a graded Poisson bracket on $C^{\infty}(\calL)$.


\subsection{Poisson algebra induced from derived brackets 
and horizontal homological functions}\label{PAfromhorizontal}
We introduce a restriction to the zero locus of functions 
$J_1, J_2 \in \Cn(T^*[n+1]\calM)$ on $\calM$
in the derived bracket,
called a \textsl{projected derived bracket} as
\begin{eqnarray}
\ssbv{J_1}{J_2} := \bsbv{\bsbv{J_1}{\Theta}}{J_2}|_{\calM},
\end{eqnarray}
where $|_{\calM}$ is the projection to the zero locus of $T^*[n+1]\calM$.
The new bracket $\ssbv{-}{-}$ is the restriction of the derived bracket 
to a Lagrangian submanifold of $T^*[n+1]\calM$.
Since the graded Poisson bracket $\bsbv{-}{-}$ is of degree $-n-1$,
the bracket $\ssbv{-}{-}$ is of degree $-n$.

If $\Theta$ is a homological function, i.e. $\bsbv{\Theta}{\Theta}=0$,
the derived bracket
$\dsbv{-}{-} = \bsbv{\bsbv{-}{\Theta}}{-}$ satisfies the graded Leibniz identity,
\begin{eqnarray*}
\dsbv{f}{\dsbv{g}{h}}&=&\dsbv{\dsbv{f}{g}}{h} 
+ (-1)^{(|f|-n)(|g|-n)}\dsbv{g}{\dsbv{f}{h}}.
\end{eqnarray*}
A projected derived bracket is given by 
$\ssbv{-}{-} = \dsbv{-}{-}|_{\calM}$.
Thus the projected derived bracket also satisfies
a graded Leibniz identity.

Now we take a specific homological function $\Theta$
called a \textit{horizontal homological function}.
\begin{definition}
A homological function $\Theta$ is called a \textit{horizontal homological function} if $\bsbv{\Theta}{J}|_{\calM} =0$ for all functions $J \in \Cn(T^*[n+1]\calM)$.
\end{definition}

The derived bracket satisfies
\begin{eqnarray}
\bsbv{\bsbv{f}{\Theta}}{g} &=& -(-1)^{(|f| - n)(|g| - n)} 
\bsbv{\bsbv{g}{\Theta}}{f} 
\nonumber \\ && 
- (-1)^{(|f| - n -1)(|g| - n -1)} \bsbv{\Theta}{\bsbv{f}{g}},
\end{eqnarray}
for any function $f, g \in \Cn(T^*[n+1]\calM)$.
If $\Theta$ is a horizontal homological function, 
the restriction of the last term to $\calM$
vanishes. Then the projected derived bracket is skew symmetric,
\begin{eqnarray}
\ssbv{f}{g} = -(-1)^{(|f| - n)(|g| - n)} \ssbv{g}{f}.
\end{eqnarray}

The derive bracket does not necessarily satisfies the Leibniz rules
(\ref{Poissonidentity2}) and (\ref{Poissonidentity3}). In fact,
\begin{eqnarray}
\bsbv{\bsbv{fg}{\Theta}}{h}
&=& f \bsbv{\bsbv{g}{\Theta}}{h}
+ (-1)^{|g|(|h|+1-n)} \bsbv{\bsbv{f}{\Theta}}{h}g
\nonumber \\
&& + (-1)^{|g|} \bsbv{f}{\Theta} \bsbv{g}{h}
+ (-1)^{(|g|+1)(|h| -n)} \bsbv{f}{h} \bsbv{g}{\Theta}
\nonumber \\ 
&=& f \bsbv{\bsbv{g}{\Theta}}{h}
- (-1)^{(|g|-n)(|h|-n)} (-1)^{|f|}
\bsbv{\bsbv{h}{\Theta}}{f}g
\nonumber \\ && 
- (-1)^{(|g|-n-1)(|h|-n-1)}
\bsbv{\Theta}{\bsbv{f}{h}}g
+ (-1)^{|g|} \bsbv{f}{\Theta} \bsbv{g}{h}
\nonumber \\
&& 
+ (-1)^{(|g|+1)(|h| -n)} \bsbv{f}{h} \bsbv{g}{\Theta},
\label{leibnizderivedbracket}
\end{eqnarray}
where $f, g, h \in C_n(T^*[n+1]\calM)$.
If $\Theta$ is a horizontal homological function, 
the projected derived bracket $\ssbv{-}{-}$
satisfies the Leibniz rule since the third term 
and the fourth term in Eq.~\eqref{leibnizderivedbracket}
vanish if the equation is restricted to $\calM$,
\begin{eqnarray}
\ssbv{fg}{h}
&=& 
f \ssbv{g}{h} 
+ 
(-1)^{|g|(|h|+1-n)} \ssbv{f}{h} g
\label{leibnizderivedbracket02}
\end{eqnarray}

From the discussion above, the following proposition is obtained.
\begin{proposition}\label{theoremcurrenthomologicalfunction}
Let $\Theta$ be a horizontal homological function. Then,
a projected derived bracket 
$\ssbv{-}{-}$ is a graded Poisson bracket of degree $-n$.
\end{proposition}
If the projected derived bracket $\ssbv{-}{-}$ is nondegenerate,
the graded symplectic structure $\omegas$ is
induced on $\calM$.



\begin{example}
In the QP-manifold $T^*[2]\calM = T^*[2]E[1]$ in Example \ref{ex052}, we take 
$E = TM$, i.e., $\calM= T[1]M$.
Take local coordinates on $T[1]M$ are $(x^i, p_i)$ of degree $(0, 1)$,
and canonical conjugate coordinates on $T^*[2]$ are
$(\xi_i, q^i)$ of degree $(2, 1)$,
where the degree one coordinate $\eta^a$ is decomposed to 
canonical conjugate coordinates $p_i$ on $T[1]$ and $q^i$ on $T^*[1]$.
The graded symplectic form is
\begin{eqnarray}
\omega = \delta x^i \wedge \delta \xi_i 
+ p_i \wedge \delta q^i,
\end{eqnarray}
which gives canonical Poisson brackets,
\begin{eqnarray}
\sbv{x^i}{\xi_j} &=& \delta^i_j,
\\
\sbv{p_i}{q^j} &=& \delta_i^j.
\end{eqnarray}
A general form of the homological function $\Theta$ of degree $3$
is 
\begin{eqnarray}
\Theta &=& \rho^i_j(x) \xi_{i} q^j 
+ \tau^{ij}(x) \xi_{i} p_j 
+ \frac{1}{3!} H_{ijk}(x) q^i q^j q^k
+ \frac{1}{2} F_{ij}^k(x) q^i q^j p_k
\nonumber \\ && 
+ \frac{1}{2} Q_i^{jk}(x) q^i p_j p_k
+ \frac{1}{3!} R^{ijk}(x) p_i p_j p_k.
\end{eqnarray}
Elements on $C_{\geq 1}(T^*[2]T[1]M)$ are
$J_0 = f(x)$ and $J_1 = v^i(x) p_i + \alpha_i(x) q^i$.
The restriction to $\calM$ is given by $q^i=\xi_i=0$.

If $\bsbv{\Theta}{J}|_{\calM} =0$ is required for all $J_0$ and $J_1$,
\begin{eqnarray}
\tau^{ij}(x) =0, Q_i^{jk}(x) = R^{ijk}(x) =0.
\end{eqnarray}
One solution of such horizontal homological functions
is 
\begin{eqnarray}
\Theta = \xi_{i} q^i + \frac{1}{3!} H_{ijk}(x) q^i q^j q^k.
\label{SCATheta}
\end{eqnarray}
$\Theta$ in Eq.~\eqref{SCATheta} satisfies
$\sbv{\Theta}{\Theta}=0$ if and only if $H = \frac{1}{3!}H_{ijk}(x) \rd x^i
\wedge \rd x^j \wedge \rd x^k$ is a closed $3$-form.
\end{example}

In our current algebra theory, we start at a 
cotangent QP-manifold bundle 
with a homological function $\Theta$,
$(\calM,T^*[n+1]\calM, \omegab, \Theta)$.
A projected derived bracket for functions 
$f, g \in \Cn(T^*[n+1]\calM)$,
$\ssbv{f}{g}= \bsbv{\bsbv{f}{\Theta}}{g}|_{\calM}$
gives a Poisson bracket for our purpose.
In this paper, we assume that the bracket $\ssbv{-}{-}$ is nondegenerate.

\subsection{Another Poisson algebras induced from derived brackets}\label{anotherBN}
The horizontal condition for a homological function is needed to make 
the projected derived bracket on $\Cn(T^*[n+1]\calM)$ a graded 
Poisson bracket.
Another method to make the projected derived bracket a graded 
Poisson bracket is to consider a subspace of $\Cn(T^*[n+1]\calM)$.
This is used for current algebras in Section \ref{CAofAKSZ}.

 We consider a function constant along the fiber,
which is a pullback of the inclusion map to zero sections,
$\iota: \calM \rightarrow T^*[n+1]\calM$.
The space of such functions is
\begin{eqnarray}
\Bn(T^*[n+1]\calM, 0):= \{f = \iota^* \tilde{f} \in \Cn(T^*[n+1]\calM)| 
\tilde{f} \in \Cn(\calM)
\}.
\end{eqnarray}
$(\Bn(T^*[n+1]\calM, 0)$ is 
a graded Poisson algebra by the bracket
$\bsbv{\bsbv{-}{\Theta}}{-}$ since $\bsbv{f}{g} =0$.
Actually, it is isomorphic to the Poisson algebra of the space of
functions of degree equal to or less than $n$ on $\calM$, 
$(\Cn(\calM), \ssbv{-}{-}= \bsbv{\bsbv{-}{\Theta}}{-}|_{\calM})$.

For a function $\alpha \in \Bn(T^*[n+1]\calM, 0)$ of degree $n+1$,
we consider the space of functions twisted by $\alpha$,
\begin{eqnarray}
\Bn(T^*[n+1]\calM, \alpha)
:= \{e^{\mathrm{ad} (\alpha)} f | f \in \Bn(T^*[n+1]\calM, 0) \}.
\label{Bntwist}
\end{eqnarray}
Since 
$\bsbv{e^{\mathrm{ad} (\alpha)} f}
{e^{\mathrm{ad} (\alpha)}g} =
e^{\mathrm{ad} (\alpha)} \bsbv{f}{g}=0$,
a pullback derived bracket $\ssbv{-}{-}= \bsbv{\bsbv{-}{\Theta}}{-}|_{\calM}$ 
is a Poisson bracket again, i.e.,
$(\Bn(T^*[n+1]\calM, \alpha)$ is a Poisson algebra.
We will see that this is the current space of AKSZ sigma models
in Section \ref{CAofAKSZ}.

Note that the Poisson algebra with the derived bracket
$\ssbv{e^{\mathrm{ad} (\alpha)} f}{e^{\mathrm{ad} (\alpha)} g}
= \bsbv{\bsbv{e^{\mathrm{ad} (\alpha)} f}
{\Theta}}{e^{\mathrm{ad} (\alpha)}g} |_{\calM}$
satisfies
\begin{eqnarray}
\ssbv{e^{\mathrm{ad} (\alpha)} f}{e^{\mathrm{ad} (\alpha)} g} 
= e^{\mathrm{ad} (\alpha)} \ssbv{f}{g}^{\alpha},
\end{eqnarray}
where $\ssbv{f}{g}^{\alpha}
:= \bsbv{\bsbv{f}{e^{- \mathrm{ad} (\alpha)} \Theta}}{g} |_{\calM}$,
because
$\bsbv{\bsbv{e^{\mathrm{ad} (\alpha)} f}{\Theta}}{e^{\mathrm{ad} (\alpha)} g}
= e^{\mathrm{ad} (\alpha)} \bsbv{\bsbv{f}{e^{- \mathrm{ad} (\alpha)} \Theta}}{g}$.





\section{Supergeometric formalism of current algebras}
\noindent
The idea is that a current algebra theory with \textit{anomalous terms}
is a Poisson algebra of functions on a QP-manifold 
$\Map(T[1]\Sigma_n,T^*[n+1]\calM)$.
The current algebra theory in this section can be regarded as
BRST-BV-AKSZ formalism of the normal current algebra theory,
which includes physical and ghost degrees of freedom.
\subsection{AKSZ-BFV constructions on 
graded cotangent bundles and projected derived brackets
}
\noindent
In this section, we construct structures on the mapping space, 
induced from a QP-structure 
on a graded cotangent bundle in the previous section.

Let $X_{n+1} = \bR \times \Sigma_n$ be an
$n+1$ dimensional spacetime manifold,
where $\Sigma_n$ is an (compact, orientable) $n$-dimensional manifold,
$\bR$ is a time direction and $\Sigma_n$ is an $n$-dimensional 
space direction.
$\calX = T[1]\Sigma_n$ is a super tangent bundle, whose fiber is 
shifted by one.
Let $(\sigma^{\mu}, \theta^{\mu})$ be local coordinate on $T[1]\Sigma_n$,
where $\sigma^{\mu}$ is a local coordinate of degree zero on $\Sigma_n$
and $\theta^{\mu}$ is a local coordinate of degree one on $T[1]$.

Assume a super differential $D$ induced from a normal de Rham differential 
$\rd$ and a compatible Berezin measure $\mu = \mu_{T[1]\Sigma_n}$.
Locally, for the de Rham differential $\rd = \rd \sigma^{\mu} \frac{\partial}{\partial \sigma^{\mu}}$, the super differential is given by
$D = \theta^{\mu} \frac{\partial}{\partial \sigma^{\mu}}$.

Let a target space $\calM$ be a graded manifold.
For $\calM$, we suppose that a cotangent bundle 
$T^*[n+1]\calM$ is a QP-manifold of degree $n+1$,
$(\calM, T^*[n+1]\calM, \omegab, \Theta, \alpha)$
with a degree $n+1$ twist function $\alpha$, 
where $\Theta$ is a horizontal homological function.
Then a projected derived bracket is defined by 
$\ssbv{-}{-} = \bsbv{\bsbv{-}{\Theta}}{-}|_{\calM}$.
Moreover, we suppose that $\ssbv{-}{-}$ is nondegenerate.

In this setting, a current algebra on the super phase space 
$\Map(T[1]\Sigma_n,\calM)$ is defined. 

A QP-structure on a graded cotangent bundle $T^*[n+1]\calM$ 
is mapped to the mapping spaces $\Map(T[1]\Sigma_n, T^*[n+1]\calM)$ by
the AKSZ construction \cite{Alexandrov:1995kv, Cattaneo:2001ys, Roytenberg:2006qz}. 
For this, a \textit{transgression map} is introduced as follows.

An evaluation map $\ev: T[1]\Sigma_n \times \Map(T[1]\Sigma_n, T^*[n+1]\calM)
\rightarrow T^*[n+1]\calM$ is defined by
${\rm ev}:(z, f) \longmapsto f(z)$,
for any
$z \in T[1]\Sigma_n$ and $f \in \Map(T[1]\Sigma_n, T^*[n+1]\calM)$.
A chain map on the space of graded differential forms,
$\mu_*:
\Omega^{\bullet}(T[1]\Sigma_n \times \Map(T[1]\Sigma_n, T^*[n+1]\calM))
\rightarrow \Omega^{\bullet}(\Map(T[1]\Sigma_n, T^*[n+1]\calM))$
 is the fiber integration with respect to a Berezin measure $\mu$ 
on $T[1]\Sigma_n$, which is defined as
$$\mu_* \omega(f)(v_1, \ldots, v_k)
 = \int_{T[1]\Sigma_n} \mu(z)
 \omega(z, f) (v_1, \ldots, v_k),$$
where $\omega$ is a graded differential form,
$v_{i}$ is a vector fields on $T[1]\Sigma_n$
and
$\int_{T[1]\Sigma_n} \mu$ is a Berezin integration on $T[1]\Sigma_n$.
The composition $\mu_* \ev^*: \Omega^{\bullet}(\calM)
\longrightarrow \Omega^{\bullet}(\Map(\calX, \calM))$
is a \textit{transgression map}.
If we define
\begin{eqnarray}
\bomegab &=& \mu_* \ev^* \omegab, 
\qquad 
\bTheta
=
\mu_* \ev^* \Theta,
\end{eqnarray}
$(\bomega_b, S_b)$ gives a QP-structure of degree $1$ on 
$\Map(T[1]\Sigma_n, T^*[n+1]\calM)$ \cite{Alexandrov:1995kv, Cattaneo:2001ys, Roytenberg:2006qz}.

We consider a projected derived bracket on the mapping space induced from the 
QP-structure on $\Map(T[1]\Sigma_n,T^*[n+1]\calM)$.
\begin{eqnarray}
\ssbv{-}{-} :=
\bsbv{\bsbv{ -}{\bTheta}}{-}|_{\Map(T[1]\Sigma_n,\calM)},
\end{eqnarray}
$\ssbv{-}{-}$ is a graded bracket on $\Map(T[1]\Sigma_n,\calM)$.
%
Here suppose that $\bTheta$ is a horizontal homological function. 
Then, applying Proposition \ref{theoremcurrenthomologicalfunction} 
to the mapping space,
the bracket $\ssbv{-}{-}$ is a graded Poisson bracket.
Note that the graded Poisson bracket $\ssbv{-}{-}$ on 
$\Map(T[1]\Sigma_n, \calM)$ 
is a normal Poisson bracket $\{-,-\}_{PB}$ because of degree $0$.
Since we assumed that the bracket $\ssbv{-}{-}$ on $\calM$
is nondegenerate, 
the Poisson bracket $\ssbv{-}{-}$ on the mapping space is nondegenerate 
and $\calM$ has a symplectic form 
$\bomegas$ induced from $\ssbv{-}{-}$.

\subsection{BFV currents and supergeometric current algebras}
\noindent
In this section, we construct a supergeometric current algebra
from a QP-structure on a mapping space.

Take a function $J$ on the space equal to or less than degree $n$, 
$\Cn(T^*[n+1]\calM)$.
Then, a functional $\calJ$ on the mapping space 
$\Map(T[1]\Sigma_n, T^*[n+1]\calM)$
is given from $J$ by the transgression map.
Precisely, the map is defined by 
$$
\calJ(\epsilon) = \mu_* \epsilon \, \ev^* J,
$$
where $\epsilon$ is a test function on $T[1]\Sigma_n$
of degree $n - |J|$.
Then, $|\calJ(\epsilon)| =0$.
Functionals $\calJ(\epsilon)$ are elements of 
$\Map(T[1]\Sigma_n, T^*[n+1]\calM)$.
the space of functionals $\calJ(\epsilon)$ for $J \in \Cn(T^*[n+1]\calM)$
consist of a Poisson algebra by the Poisson bracket $\ssbv{-}{-}$,
since $\Cn(T^*[n+1]\calM)$ is closed under the bracket $\ssbv{-}{-}$
and the transgression map is commutative with the graded Poisson bracket,
\begin{eqnarray}
\mu_* \ev^* \bsbv{f}{g} = \bsbv{\mu_* \ev^* f}{\mu_* \ev^* g}.
\label{Poissontransg}
\end{eqnarray}

However, the Poisson algebra of $\calJ(\epsilon)$ does not give
physical current algebras yet.
We consider twist of functionals $\calJ(\epsilon)$
on $\Map(T[1]\Sigma_n,T^*[n+1]\calM)$ by a particular twisting function. 
Proper twist of $\calJ(\epsilon)$ gives current algebras.

For a differential $D$ on $\calX = T[1]\Sigma_n$,
we denote by the same notation $D$ 
the vector field on $\Map(T[1]\Sigma_n, \calM)$
of degree $1$ induced by $D$.
Take local coordinates on $T[1]\Sigma_n$ as 
$(\sigma^{\mu}, \theta^{\mu})$ of degre $(0, 1)$.
We define a twisting function by 
$\bvarthetas = \iota_{D} \mu_* \ev^* \proj^* \varthetas$,
where $\varthetas$ is the canonical $1$-form for a symplectic form
$\omegas$ on $\calM$ such that $\omegas = - \delta \varthetas$.
The functional $\bvarthetas$ 
is the same as the kinetic term of the BV action functional
in the AKSZ sigma model on $\Map(T[1]\Sigma_n,\calM)$
and satisfies 
$\bsbv{\bvarthetas}{\bvarthetas}=0$ and $\ssbv{\bvarthetas}{\bvarthetas}=0$.

For any function $J$ on $T^*[n+1]\calM$,
we consider 
twisting of the functional $\calJ(\epsilon)$
with respect to $\bvarthetas$,
\begin{eqnarray}
e^{\mathrm{ad}(\bvarthetas)}\calJ(\epsilon)
= e^{\mathrm{ad}(\bvarthetas)}\mu_*\epsilon \, \ev^*J.
\end{eqnarray}
where 
$e^{\mathrm{ad}(\bvarthetas)} f = f + \bsbv{f}{\bvarthetas}
+ \frac{1}{2} \bsbv{\bsbv{f}{\bvarthetas}}{\bvarthetas}
+ \frac{1}{3!} \bsbv{\bsbv{\bsbv{f}{\bvarthetas}}{\bvarthetas}}{\bvarthetas}
+ \cdots$.

We define a current in our formalism.
\begin{definition}
A BFV current $\bJ(\epsilon)$ on $\Map(T[1]\Sigma_n,\calM)$
with respect to a function $J \in \Cn(T^*[n+1]\calM)$ is 
the restriction to $\Map(T[1]\Sigma_1,\calM)$ of the transgression of $J$, 
i.e., it is defined by 
\begin{eqnarray}
\bJ(\epsilon)=
\bJ_J(\epsilon) := e^{\mathrm{ad}(\bvarthetas)}\calJ(\epsilon)
|_{\Map(T[1]\Sigma_1,\calM)}
.
\end{eqnarray}
The space of BFV currents are denoted by $\CAn(T^*[n+1]\calM)$.
\end{definition}
The Poisson bracket of currents $\bJ_1(\epsilon_1)$
and $\bJ_2(\epsilon_2)$ associated 
to two functions $J_1, J_2 \in \Cn(T^*[n+1]\calM)$ is calculated
from the derived bracket on $T^*[n+1]\calM$.

We prove it by taking the Darboux coordinate on $T^*[n+1]\calM$,
$(p_{I^{(i)}}, q^{I^{(n+1-i)}})$ 
of degree $(i, n+1-i)$, where $p_{I^{(i)}}$ is a coordinate on $\calM$
and $q^{I^{(n+1-i)}}$ is a coordinate on the fiber $T^*[n+1]$, 
where $i=0,\cdots, \floor{\frac{n+1}{2}}$.
The canonical graded symplectic structure is
$\omegab = 
\sum_{i=0}^{n} 
(-1)^{i(n-i)}
\delta p_{I^{(i)}} \wedge \delta q^{I^{(n+1-i)}}$.
Then, the Poisson bracket $\bsbv{-}{-}$ is
$
\bsbv{p_{I^{(i)}}}{q^{J^{(j)}}} = \delta_{I^{(i)}}{}^{J^{(n+1-j)}}
\delta_{i}{}^{j}$.
For local coordinates on the mapping space
corresponding to $(p_{I^{(i)}}, q^{I^{(n+1-i)}})$, 
we take superfields
$(\bp_{I^{(i)}}, \bq^{I^{(n+1-i)}})$, which
are sections of $T[1]\Sigma_n \otimes
\bx^* (T^*[n+1]\calM_{(i)})$ and
$T[1]\Sigma_n \otimes
\bx^* (T^*[n+1]\calM_{(n+1-i)})$, 
where $T^*[n+1]\calM_{(i)}$ is a degree $i$ part of 
$T^*[n+1]\calM$.
Here the degree 0 coordinate is denoted by
$\bx^I = \bp_{I^{(0)}}$.
$\bx^I$ is regarded as a map 
$\bx^I : T[1]\Sigma_n \rightarrow M$.
The graded symplectic structure on the mapping space is 
$\bomegab = 
\sum_{i=0}^{n} \int_{T[1]\Sigma_n} 
(-1)^{i(n-i)}
\delta \bp_{I^{(i)}} \wedge \delta \bq^{I^{(n+1-i)}}$.
The restriction to $\Map(T[1]\Sigma_1,\calM)$ is
given by $\bq^{I^{(n+1-i)}}=0$ for all $i$.

First, we calculate the Poisson bracket of untwisted functional
$\calJ_{J_1}(\epsilon_1)$ and $\calJ_{J_2}(\epsilon_2)$.
Since a transgression map $\mu_* \ev^*$ is a Poisson map \eqref{Poissontransg}
from $T^*[n+1]\calM$ to $\Map(T[1]\Sigma_n, T^*[n+1]\calM)$,
we obtain
\begin{eqnarray}
\{\calJ_{J_1}(\epsilon_1),\calJ_{J_2}(\epsilon_2)\}_{PB}
&=&
\ssbv{\mu_* \epsilon_1 \ev^* J_1}{\mu_* \epsilon_2 \ev^* J_2}
\nonumber \\
&=& 
\mu_* \epsilon_1 \epsilon_2 
\ev^* \ssbv{J_1}{J_2}
\nonumber \\
&=& 
\mu_* \epsilon_1 \epsilon_2
\ev^* \bsbv{\bsbv{J_1}{\Theta}}{J_2}
|_{\Map(T[1]\Sigma_n,\calM)}.
\label{untwisted01}
\end{eqnarray}

Next we calculated the Poisson bracket of functionals twisted by $\bvarthetas$.
$\bvarthetas = \iota_{D} \mu_* \ev^* \vartheta_s$ is a Hamiltonian 
function for the vector field $D$ with respect to $\omegas$, i.e.,
$X_{\bvarthetas} = D$, where the Hamiltonian vector field $X_f$ for a function 
$f$ is defined by 
$\iota_{X_f} \omega = - \delta f$.
From $X_{\bvarthetas} = D$, we have
\begin{eqnarray}
\ssbv{\bvarthetas}{\mu_* \ev^* f}
&=&
(-1)^{|\bvarthetas|} \iota_{\hat{D}} \iota_{X_{\mu_* \ev^* f}} \bomega
\nonumber \\ &=&
- \iota_{\hat{D}} \mu_* \ev^* \delta f
= \left( \int \rd^{n+1}\sigma \rd^{n+1}\theta \bbd f(\sigma, \theta) \right),
\label{smallHamiltonian}
\end{eqnarray}
for $f \in C^{\infty}(\calM)$.

Using \eqref{smallHamiltonian}, we can calculate the following formula 
fore Poisson bracket on $\Map(T[1]\Sigma_n,\calM)$, 
$\sbv{-}{-}_{PB}= \ssbv{-}{-}$
and bracket $\Map(T[1]\Sigma_n, T^*[n+1]\calM)$ and $\bsbv{-}{-}$:
\begin{eqnarray}
&& \sbv{e^{\mathrm{ad}(\bvarthetas)} \bq^{J^{(n+1-j)}}(\sigma, \theta)}
{\bp_{I^{(i)}}(\sigma^{\prime}, \theta^{\prime})}_{PB}
\nonumber \\ &=& 
\ssbv{e^{\mathrm{ad}(\bvarthetas)} \bq^{J^{(n+1-j)}}(\sigma, \theta)}
{\bp_{I^{(i)}}(\sigma^{\prime}, \theta^{\prime})}
\nonumber \\ &=& 
\ssbv{\bq^{J^{(n+1-j)}}(\sigma, \theta)}
{\bp_{I^{(i)}}(\sigma^{\prime}, \theta^{\prime})}
+ 
\delta_{I^{(i)}}{}^{J^{(n+1-j)}}
\delta_{i}{}^{j}
\bbd (\delta(\sigma - \sigma^{\prime})
\delta(\theta - \theta^{\prime}))
\nonumber \\
&=& 
\bsbv{\bsbv{\bq^{J^{(n+1-j)}}(\sigma, \theta)}{S}}{\bp_{I^{(i)}}(\sigma^{\prime}, \theta^{\prime})}
+ \bsbv{q^{J^{(n+1-j)}}}
{p_{I^{(i)}}}
\bbd (\delta(\sigma - \sigma^{\prime})
\delta(\theta - \theta^{\prime})),
\label{Poissb11}
\end{eqnarray}
using local coordinates $(\sigma^{\mu}, \theta^{\mu})$ 
on $T[1]\Sigma_n$.
By applying the same equation to general BFV currents,
we obtain the general formula,
\begin{eqnarray}
\{\bJ_{J_1}(\epsilon_1),\bJ_{J_2}(\epsilon_2)\}_{PB}
&=&
\left(
e^{\mathrm{ad}(\bvarthetas)} \mu_* \epsilon_1 \epsilon_2
\ev^* \bsbv{\bsbv{J_1}{\Theta}}{J_2}
\right.
\nonumber \\
&& \left.
-e^{\mathrm{ad}(\bvarthetas)}
\mu_*
(D \epsilon_1) \epsilon_2
\ev^* \bsbv{J_1}{J_2}
\right)
|_{\Map(T[1]\Sigma_n,\calM)}
\label{supercurrentcommutation}
\\
&=& -\bJ_{[J_1,J_2]_D}
(\epsilon_1 \epsilon_2)
- e^{\mathrm{ad}(\bvarthetas)} 
\mu_*
(D \epsilon_1) \epsilon_2
\ev^* \bsbv{J_1}{J_2}
|_{\Map(T[1]\Sigma_n,\calM)}
.
\nonumber
\end{eqnarray}
The second term is in Equation (\ref{supercurrentcommutation})
comes from terms proportional to a derivative of the delta function
in Equation \eqref{Poissb11}.
Here we use $\bsbv{\bvarthetas}{\bvarthetas}=0$ and $\ssbv{\bvarthetas}{\bvarthetas}=0$, and
integrating by parts for terms with test functions.

In general, 
because of twisting by $\bvarthetas$, closedness of currents $\bJ$ under 
the Poisson bracket, 
$\ssbv{-}{-} = \bsbv{\bsbv{-}{\Theta}}{-}|_{\calM}$,
is broken.
Breaking terms are the second terms
in \eqref{supercurrentcommutation} called anomalous terms.
As a result, the Poisson bracket of $\bJ_1$ and $\bJ_2$ 
is given by the Poisson bracket $\bsbv{-}{-}$ and
its derived bracket $\bsbv{\bsbv{-}{\Theta}}{-}$
on the QP-manifold $T^*[n+1]\calM$ 
as in Equation (\ref{supercurrentcommutation}).

We summarize our formulas of the current algebra.
\begin{theorem}\label{supercurrentalgebratheorem}
Let $\Theta$ be a horizontal homological function on $T^*[n+1]\calM$.
For BFV currents $\bJ_{J_1}$ and $\bJ_{J_2}$ associated to 
functions $J_1$, $J_2$ in $\Cn(T^*[n+1]\calM)$,
the Poisson bracket
is given by
\begin{eqnarray}
\{\bJ_{J_1}(\epsilon_1),\bJ_{J_2}(\epsilon_2)\}_{PB}
&=&
\left(
e^{\mathrm{ad}(\bvarthetas)} \mu_* \epsilon_1 \epsilon_2
\ev^* \bsbv{\bsbv{J_1}{\Theta}}{J_2}
\right.
\nonumber \\
&& \left.
-e^{\mathrm{ad}(\bvarthetas)}
\mu_*
(D \epsilon_1) \epsilon_2
\ev^* \bsbv{J_1}{J_2}
\right)
|_{\Map(T[1]\Sigma_n,\calM)}
\label{supercurrentcommutation02}
\\
&=& -\bJ_{[J_1,J_2]_D}
(\epsilon_1 \epsilon_2)
- e^{\mathrm{ad}(\bvarthetas)} 
\mu_*
(D \epsilon_1) \epsilon_2
\ev^* \bsbv{J_1}{J_2}
|_{\Map(T[1]\Sigma_n,\calM)}
.
\nonumber
\end{eqnarray}
where $\epsilon_i$ are test functions for $J_i$ of degree $n-|J_i|$
on $\Map(T[1]\Sigma_n, T^*[n+1]\calM)$.
\end{theorem}

%

The anomalous terms vanish if $J_1$ and $J_2$ commute,
$\bsbv{J_1}{J_2}=0$. Thus, we have
Let $A = Comm_{\leq n}(T^*[n+1]\calM)$ be a subspace of
commutative functions in $\Cn(T^*[n+1]\calM)$ 
under the graded Poisson bracket $\sbv{-}{-}$ on $T^*[n+1]\calM$.
$A = Comm_{\leq n}(T^*[n+1]\calM)$ is not unique in $\Cn(T^*[n+1]\calM)$.
\begin{corollary}
Then, $A$ gives
a current algebra without anomalous terms from the formula
\eqref{supercurrentcommutation02}.
\end{corollary}
We can see how this theory works from the viewpoint of Lagrangian submanifolds.
We do not use the restriction of the simple transgression
of functions on $T^*[n+1]\calM$ 
in the definition of BFV currents in $\Map(T[1]\Sigma_n,\calM)$.
Functions on $\Map(T[1]\Sigma_n,T^*[n+1]\calM)$ are twisted
by $\bvarthetas$ before the restriction to $\Map(T[1]\Sigma_n,\calM)$.
By this procedure, nontrivial anomalous terms are nonzero 
in the current algebra.
We have another equivalent description of the supergeometric current algebra.
$\Map(T[1]\Sigma_n,\calM)$ is a trivial Lagrangian submanifold
of $\Map(T[1]\Sigma_n,T^*[n+1]\calM)$.
If a Lagrangian submanifold $\Map(T[1]\Sigma_n,\calM)$ 
is twisted by $\bvarthetas$,
it is transformed to a different Lagrangian submanifold
$\Map(T[1]\Sigma_n, \widetilde{\calM})$.
The new Lagrangian submanifold $\Map(T[1]\Sigma_n, \widetilde{\calM})$
is also equipped with a Poisson bracket
given by the projected derived bracket, 
and is isomorphic to $\Map(T[1]\Sigma_n,\calM)$ as a P-manifold.
If unctionals on $\Map(T[1]\Sigma_n,T^*[n+1]\calM)$
are restricted to this twisted Lagrangian submanifold,
Poisson brackets have anomalous terms.

A Lagrangian submanifold in supergeometry is reinterpreted as 
a generalized Dirac structure in normal geometry. 
Relations of the Dirac structure and current algebras were discussed in 
\cite{Alekseev:2004np}.

\subsection{Physical current algebras from supergeometric current algebras}
\noindent
Physical current algebras are obtained from restriction of BFV current algebras in the previous section.

For it, we introduce second degree, the \textsl{form degree} $\deg f$
for a function $f$ on $T[1]\Sigma_n$.
It is, by definition, zero on $\Sigma_n$ and
one on the $T[1]$ direction.
If we take a local coordinate $(\sigma^{\mu}, \theta^{\mu})$
on $T[1]\Sigma_n$,
the form degree of $(\sigma^{\mu}, \theta^{\mu})$ is $(0, 1)$.
For example, $f = \frac{1}{k!} \theta^{\mu_1} \cdots \theta^{\mu_k}
f_{\mu_1 \cdots \mu_k}(\sigma)$ is the form degree $k$.
For a general function $f$ on
$\Map(T[1]\Sigma_n, T^*[n+1]\calM$
with definite degree, degree minus form degree,
$\gh \ f = |f| - \deg f$,
is called a \textsl{ghost number}.

We expand a superfield of degree $i$ by $\theta$ as
\begin{eqnarray}
\bPhi_{(i)}(\sigma, \theta)
= \sum_{k, \mu(k)} \bPhi_{(i)}^{(k)}(\sigma, \theta)
= \sum_{k, \mu(k)} \frac{1}{k!}
\theta^{\mu_1} \cdots \theta^{\mu_k}
\Phi^{(k)}_{(i)\mu_1 \cdots \mu_k}(\sigma).
\nonumber
\end{eqnarray}
Then, the term $\frac{1}{k!}
\theta^{\mu_1} \cdots \theta^{\mu_k}
\Phi^{(k)}_{(i)\mu_1 \cdots \mu_k}(\sigma)$
has the form degree $k$ and the ghost number $i-k$.
A classical physical field is the coefficient of the ghost number 
zero term of $\bPhi_{(i)}$,
$\Phi^{(i)}_{(i)\mu_1\cdots \mu_k}(\sigma)
= \bPhi_{(i)}{}_{\mu_1\cdots \mu_k}(\sigma)|_{\gh \Phi =0}$.
We denote 
the restriction of a superfield function $f$ 
to the corresponding classical field term
by $f_{cl}$, i.e.,
\begin{eqnarray}
\bPhi_{(i)}{}_{cl}(\sigma, \theta)
:= \frac{1}{i!}
\theta^{\mu_1} \cdots \theta^{\mu_i}
\Phi^{(i)}_{(i)\mu_1 \cdots \mu_i}(\sigma),
\label{classicalcomp}
\end{eqnarray}

Recall Darboux coordinate superfields
$(\bp_{(i)}(\sigma, \theta), \bq_{(n+1-i)}(\sigma, \theta))$
on $\Map(T[1]\Sigma_n, T^*[n+1]\calM)$,
where $\bp_{(i)}$ is an element of
$\Map(T[1]\Sigma_n, \calM)$
of degree $i$ and
$\bq_{(n+1-i)}$ is an element of
$\Map(T[1]\Sigma_n, T^*[n+1])$
of degree $n+1-i$,
$(i=0, 1, \cdots, n)$.
Since we suppose that $S$ is horizontal,
a Poisson bracket $\sbv{-}{-}_{PB}$ as the projected derived bracket
from the derived bracket constructed from $\bomegab$ and $S$.
By restricting $\sbv{-}{-}_{PB}$
to the ghost number zero part of superfields $\bp_{(i)}^{(i)}$,
the Poisson brackets of physical fields are obtained.

The restriction of $\bJ$ to ghost number zero components of superfields 
\eqref{classicalcomp} is denoted by
$\bJ{}_{cl} 
= \int_{\Sigma_n} \! 
\cepsilon_{(n-i)}^{(n-i)} (\sigma)
\wedge \cJ^{(i)}_{(i)}(\sigma)$.
Here test functions $\cepsilon_{(n-i)}^{(n-i)} (\sigma)$ are also the ghost number zero part of super test functions $\epsilon$.

\begin{theorem}
The classical current algebra is the Poisson algebra of
the ghost number zero components of superfields of
the supergeometric current algebra. The concrete formula is
\begin{eqnarray}
\{\bJ_{J_1}{}_{cl}(\epsilon_1),\bJ_{J_2}{}_{cl}(\epsilon_2)\}_{P.B}
&=& \left(-\bJ_{[J_1,J_2]_D}
(\epsilon_1 \epsilon_2)
\right.
\nonumber \\
&&
\left.
- e^{\mathrm{ad}(\bvarthetas)} \mu_*
(D \epsilon_1) \epsilon_2
\ev^* \bsbv{J_1}{J_2}
\right)|_{\Map(T[1]\Sigma_n,\calM)}
{}_{cl},
\label{classicalcurrentcommutation}
\end{eqnarray}
where ${}_{cl}$ is components of classical parts of superfields.
\end{theorem}
In the next subsection, we use this formula to obtain
concrete current algebras.

\subsection{Examples}
In this section, examples of supergeometric current algebras are listed.

\begin{example}[Alekseev-Strobl current algebras]\label{AScurrentalgebra}
\cite{Alekseev:2004np}
Let us consider the loop phase space, $\Map(S^1,T^*M)$.
The loop phase space is extended to the super loop space,
$\Map(T[1]S^1,T^*[1]M)$.
Our current algebra is defined on this super loop phase space.

Start at the target supermanifold $T^*[1]M$.
For a supermanifold $\calM = T^*[1]M$, we consider
a cotangent bundle $T^*[2]\calM = T^*[2]T^*[1]M$.
Following Example \ref{examplepoisson}, we consider 
QP-structure on $T^*[2]T^*[1]M$ with 
the canonical symplectic structure $\omega_{T^*[2]T^*[1]M}$.
Here we consider the horizontal homological function 
$\Theta = \xi_I q^I + \frac{1}{3!} H_{IJK}(x) q^I q^J q^K$.
Then, \\
$(T^*[2]T^*[1]M, \omega_{T^*[2]T^*[1]M}, \Theta)$ is a QP-manifold.

By using the AKSZ construction,
a QP-structure on the mapping space $\Map(T[1]S^1,T^*[2]T^*[1]M)$ is induced
from the QP-structure on $T^*[2]T^*[1]M$.
Superfields are local coordinates in the mapping space. 
They are
$\bx:T[1]S^1 \rightarrow M$,
$\bp \in \Gamma(T[1]S^1 \otimes \bx^*(T^*[1]M))$,
$\bq \in \Gamma(T[1]S^1 \otimes \bx^*(T^*[2]T^*[1]M))$
and
$\bbxi \in \Gamma(T[1]S^1 \otimes \bx^*(T^*[2]M))$.
Expansions with respect to $\theta$ are
\begin{eqnarray}
\bx^i &=& x^i(\sigma) + \theta x^i_{(1)}(\sigma),
\\
\bp_i &=& p_{(0)i} (\sigma) + \theta p_i(\sigma),
\\
\bq^i &=& q^i_{(0)}(\sigma) + \theta q^i_{(1)}(\sigma),
\\
\bbxi_i &=& \xi_{(0)i} (\sigma) + \theta \xi_{(1)i}(\sigma).
\end{eqnarray}
The symplectic structure on 
$\Map(T[1]S^1,T^*[2]T^*[1]M)$
is
\begin{eqnarray}
\bomega_{T^*[2]T^*[1]M}
= 
\int_{T[1]S^1}
\rd \sigma \rd \theta
\,
\left(\delta \bx^{I} \wedge \delta \bbxi_{I}
+ \delta \bp_{I} \wedge \delta \bq^{I}
\right).
\end{eqnarray}
The homological function is
\begin{eqnarray}
S = \int_{T[1]S^1}
\rd \sigma \rd \theta
\,
\left( \bbxi_I \bq^I + \frac{1}{3!} H_{IJK}(\bx) \bq^I \bq^J \bq^K \right).
\end{eqnarray}
where boldface letters are superfields corresponding to local coordinates
on $T^*[2]T^*[1]M$.
The Berezin measure on $T[1]S^1$ is denoted by $\mu = \rd \sigma \rd \theta$.

The projected derived bracket $\ssbv{-}{-}$ is a Poisson bracket
of degree zero $\sbv{-}{-}_{PB}$ on the super phase space 
$\Map(T[1]S^1,T^*[1]M)$. The Poisson bracket is
\begin{eqnarray}
\sbv{\bx^{I}(\sigma, \theta)}{\bp_{J}(\sigma^{\prime}, \theta^{\prime})}_{PB} 
= \delta^I_J \delta(\sigma-\sigma^{\prime}) \delta(\theta-\theta^{\prime}),
\end{eqnarray}
and other canonical Poisson brackets are zero.
Then, the following symplectic structure on $\Map(T[1]S^1, T^*[1]M)$ 
is induced,
\begin{eqnarray}
\bomega_{T^*[1]M}
= 
\int_{T[1]S^1}
\rd \sigma \rd \theta
\,
\delta \bx^{I} \wedge \delta \bp_{I}.
\end{eqnarray}
The canonical $1$-from $\bvarthetas$ is
\begin{eqnarray}
\bvartheta_{T^*[1]M}
&=& \iota_{D} \mu_* \ev^* \vartheta_{T^*[1]M}
= \int_{T[1]S^1}
\rd \sigma \rd \theta
\
\bp_I \bbd \bx^I.
\end{eqnarray}

Next, we consider the space of BFV currents,
i.e., the space of the restriction of twisted functionals
\begin{eqnarray}
\CAone(T^*[2]T^*[1]M)) = \{\bJ 
| J \in \Cn(T^*[2]T^*[1]M), \epsilon \in C^{\infty}(T[1]S^1)
\},
\end{eqnarray}
where $\bJ
= e^{\mathrm{ad}(\bvarthetas)} \calJ|_{\Map(T[1]S^1, T^*[1]M)}
= e^{\mathrm{ad}(\bvarthetas)} \mu_* \epsilon \, \ev^* J |_{\Map(T[1]S^1, T^*[1]M)}$.
\\
%
%
%
Elements of $C_{\leq 1}(T^*[2]T^*[1]M)$ are functions of degree $0$
and of degree $1$. 
We can write them as
\begin{eqnarray}
&& J_{(0)(f)} = f(x),
\nonumber \\
&& J_{(1)(u,a)} = a_I(x) q^I + u^I(x) p_I,
\end{eqnarray}
where $f, a, u$ are local functions of $x$.
Corresponding currents
are
\begin{eqnarray}
&& \bJ_{(0)(f)}
= \int_{T[1]S^1} 
\rd \sigma \rd \theta \, 
\epsilon_{(1)} f(\bx),
\nonumber \\
&&
\bJ_{(1)(u,a)}
= \int_{T[1]S^1} \rd \sigma \rd \theta \, 
\epsilon_{(0)} (a_I(\bx) \bbd \bx^I + u^I(\bx) \bp_I),
\end{eqnarray}
where $\epsilon_{(i)}$ is a test function of degree $i$.
Note that $\bq^I$ is replaced to $\bbd \bx^I$
by twist $e^{\mathrm{ad}(\bvarthetas)}$.
Poisson brackets of these currents are computed by
the formula \ref{supercurrentcommutation02}
 as follows:
\beqa
    \{\bJ_{(0)(f)}(\epsilon),
\bJ_{(0)(g)}^{\prime}(\epsilon^{\prime})
\}_{PB}
&=&
\left( - e^{\mathrm{ad}(\bvarthetas)} \mu_* \epsilon \epsilon^{\prime}
\ev^* \bsbv{\bsbv{J_{(0)(f)}}{\Theta}}{J_{(0)(g)}^{\prime}}\right.
\nonumber \\ &&
\left. -e^{\mathrm{ad}(\bvarthetas)} \mu_*
(D \epsilon) \epsilon^{\prime}
\ev^* \bsbv{J_{(0)(f)}}{J_{(0)(g)}}
\right)|_{\Map(T[1]S^1, T^*[1]M)}
=        0,\nom\\
     \{\bJ_{(1)(u,a)}(\epsilon),
\bJ_{(0)(g)}^{\prime}(\epsilon^{\prime})\}_{PB}
&=&
\left( - e^{\mathrm{ad}(\bvarthetas)} \mu_* \epsilon \epsilon^{\prime}
\ev^* \bsbv{\bsbv{J_{(1)(u,a)}}{\Theta}}{J_{(0)(g)}^{\prime}}\right.
\nonumber \\ &&
\left. -e^{\mathrm{ad}(\bvarthetas)} \mu_*
(D \epsilon) \epsilon^{\prime}
\ev^* \bsbv{J_{(1)(u,a)}}{J_{(0)(g)}}
\right)|_{\Map(T[1]S^1, T^*[1]M)}
\nom\\
&=& -u^{I}\frac{\del \bJ_{(0)(g)}^{\prime}}
{\del \bx^I}(\epsilon \epsilon^{\prime})
,\nom\\
    \{\bJ_{(1)(u,a)}(\epsilon),
\bJ_{(1)(v,b)}(\epsilon^{\prime})\}_{PB}
    &=&
\left( - e^{\mathrm{ad}(\bvarthetas)} \mu_* \epsilon \epsilon^{\prime}
\ev^* \bsbv{\bsbv{J_{(1)(u,a)}}{\Theta}}{J_{(1)(v,b)}^{\prime}}
\right.
\nonumber \\ &&
\left. -e^{\mathrm{ad}(\bvarthetas)} \mu_*
(D \epsilon) \epsilon^{\prime}
\ev^* \bsbv{J_{(1)(u,a)}}{J_{(1)(v,b)}}
\right)|_{\Map(T[1]S^1, T^*[1]M)}
\nonumber \\
    &=&
-\bJ_{(1)(\courant{(u,a)}{(v,b)})}(\epsilon \epsilon^{\prime})
\nonumber \\
&&
- \int_{T[1]S^1} 
\!\!\!\!\!\!\!\!
\rd \sigma \rd \theta \,
\bbd \epsilon_{(0)} \epsilon^{\prime}_{(0)}
\bracket{(a_I(\bx), u^I(\bx))}{(b_I(\bx), v^{I}(\bx))},
\label{2dalekseevstrtoblsupercommutation}
\eeqa
where
$\bJ_{(0)(g)}^{\prime} = \int_{T[1]S^1} \rd \sigma \rd \theta \, \epsilon_{(1)} g(\bx)$
and
$\bJ_{(1)(v,b)}^{\prime} =
\int_{T[1]S^1} \rd \sigma \rd \theta \, \epsilon_{(0)} (b_I(\bx) \bbd \bx^I + v^I(\bx) \bp_I)$.
$\courant{(u,a)}{(v,b)}$ is the Dorfman bracket on
$TM \oplus T^*M$ explained in Example \ref{examplepoisson}.
The classical current algebra is the ghost number zero
components of superfields in
the equations (\ref{2dalekseevstrtoblsupercommutation}):
\beqa
    \{\bJ_{(0)(f)}{}_{cl}(\epsilon),
\bJ_{(0)(g)}^{\prime}{}_{cl}(\epsilon^{\prime})
\}_{PB}
&=&
0,\nom\\
     \{\bJ_{(1)(u,a)}{}_{cl}(\epsilon),
\bJ_{(0)(g)}^{\prime}{}_{cl}(\epsilon^{\prime})\}_{PB}
&=& -u^{I}(\sigma)\frac{\del \bJ_{(0)(g)}^{\prime}{}_{cl}}
{\del x^I}(\epsilon \epsilon^{\prime})
,\nom\\
    \{\bJ_{(1)(u,a)}{}_{cl}(\epsilon),
\bJ_{(1)(v,b)}{}_{cl}(\epsilon^{\prime})\}_{PB}
    &=&
- \bJ_{(1)(\courant{(u,a)}{(v,b)})}{}_{cl}(\epsilon \epsilon^{\prime})
\label{2dalekseevstrtoblcommutation}
\\
&&
- \int_{S^1} \rd \sigma (\partial_{\sigma} \epsilon_{(0)}
\epsilon^{\prime}_{(0)}
\bracket{(a_I(x(\sigma)), u^I(x(\sigma)))}{(b_I(x(\sigma)), v^{I}(x(\sigma)))},
\nonumber
\eeqa
where
$\bJ_{(0)(f)}{}_{cl}
= \int_{S^1} \rd \sigma \cepsilon_{(1)}(\sigma) f(x(\sigma))$,
$\epsilon_{(1)}(\sigma, \theta) = \theta \cepsilon_{(1)}(\sigma)$,
and
\\
$\bJ_{(1)(u,a)}{}_{cl}
= \int_{S^1} \rd\sigma \epsilon_{(0)}(\sigma)
(a_I(x(\sigma)) \partial_{\sigma} x^I(\sigma) + u^I(x(\sigma)) p_I(\sigma))$.
%
This coincides with the generalized current algebra described
in \cite{Alekseev:2004np}. 
\end{example}

\begin{example}[Kac-Moody Algebra]
As a special case of the previous example,
take $M$ to be a Lie group $G$.
Then the extended phase space is $\Map(T[1]S^1,T^*[1]G)$.
We choose a QP-manifold $(T^*[2]T^*[1]G, \omega_{T^*[2]T^*[1]G}, \Theta)$
and the graded symplectic form and the homological function $\Theta$ 
are the same forms as in Example \ref{AScurrentalgebra}.
Here $H_{IJK}$ is a structure constant of the corresponding Lie algebra.

Then the BFV currents give the supergeometric formalism
of the Kac-Moody current algebra.
The classical part $\cJ$ coincides with
currents of the Hamiltonian formalism of the WZW model
\cite{Giusto:2001sn}.
\end{example}

\begin{example}[Current algebras in two dimensions]\cite{Ikeda:2011ax}
We can consider higher dimensional generalizations of 
Alekseev-Strobl current algebras in Example \ref{AScurrentalgebra}.

Let $X = \Sigma_2 \times \bR$ be a three dimensional spacetime.
The two-dimensional manifold $\Sigma_2$ is the space direction.
$\sigma^{\mu}, (\mu =1,2)$ is a local coordinate on $\Sigma_2$.
Let us consider the space,
$\Map(\Sigma_2, T^*M) \oplus \Map(T\Sigma_2, T^*E)$,
where $E$ is a vector bundle over a manifold $M$.
This is the phase space for physical models in $3$ dimensions.
Local coordinates are $(x^I(\sigma), p_{I}(\sigma), q^A(\sigma))$,
where $I$ is the index in $M$ and $A$ is the index of the fiber of $E$.
Here $(x^I(\sigma), p_{I}(\sigma))$ are local coordinates 
of $\Map(\Sigma_2, T^*M)$, and a $1$-form taking a value in $X^*E$, 
$q^A(\sigma) = \rd \sigma^{\mu} q_{\mu}^A(\sigma))$, is regarded as 
an element of $\Map(T\Sigma_2, T^*E)$.
Let us consider a super phase space 
$\Map(T[1]\Sigma_2, T^*[2]E[1])$ as an extension of the above classical 
phase space.

We introduce the cotangent bundle $T^*[3]\calM = T^*[3]T^*[2]E[1]$
as the target space of an auxiliary phase space.
$T^*[3]\calM$ has a canonical graded symplectic form 
$\omega_{T^*[3]\calM}$ of degree $3$.
Take local coordinates $(x^I, q^A, p_I)$
of degree $(0,1,2)$ on $T^*[2]E[1]$,
and conjugate Darboux coordinates $(\xi_I, \eta^A, \chi^I)$ of degree
$(3,2,1)$ on $T^*[3]$.
A graded symplectic structure on $T^*[3]\calM$ is
\beq
\omega_{T^*[3]\calM} = \delta x^I\wedge\delta \xi_I
    -\delta q^A\wedge\delta (k_{AB}\eta^{B})
+\delta p_I\wedge\delta \chi^I,
\eeq
where $k_{AB}$ is a fiber metric on $E$\footnote{We can take $k_{AB} = \delta_{AB}$.}.

We consider a Q-structure function,
\beq
\Theta=  \chi^I \xi_I
          +\frac{1}{2}k_{AB}\eta^A\eta^B+
\frac{1}{4!}
H_{IJKL}(x) \chi^I \chi^J \chi^K \chi^L.
\eeq
$\Theta$ is a compatible homological function if $H$ is a closed $4$-form.

A graded symplectic form
$\bomega_{T^*[3]\calM}$ of degree $1$ is induced as
\beq
    \bomega_{T^*[3]\calM}
=\mu_* \ev^* \omegab
=\int_{T[1] \Sigma_2}
\rd \sigma^2 \rd^2 \theta \,
(\delta \bx^I\wedge\delta\vxi_I
    -\delta\bq^A\wedge\delta (k_{AB} \veta^{B})
+\delta \bp_I\wedge\delta \vchi^I),
\eeq
Boldface superfields are ones for corresponding local coordinates 
in $T^*[3]\calM$.
%

In order to construct currents, 
consider the space of functions of degree equal to, or less than two, 
$C_2(T^*[3]T^*[2]E[1])
= \{f \in C^{\infty}(T^*[3]T^*[2]E[1]) | |f| \leq 2 \}$.
Elements of degree $0$, $1$ and $2$ are written by local coordinates as
\begin{eqnarray}
J_{(0)(f)} &=& f(x),
\nonumber \\
J_{(1)(a, u)} &=&
a_{I}(x) \chi^I + u_{A} (x) q^A,
\nonumber \\
J_{(2)(G,K,F,B,E)} &=&
    G^I(x) p_{I} + K_A(x) \eta^A + \frac{1}{2}F_{AB}(x) q^A q^B
+\frac{1}{2}B_{IJ}(x)\chi^I \chi^J
\nonumber \\&&
+ E_{AI}(x) \chi^I q^A.
\end{eqnarray}
Here all coefficients are some functions of $x$.

Next we take the space of supergeometric currents, \\
$\calC\calA_2(T^*[2]E[1])) = \{\bJ 
| J \in C_2(T^*[3]T^*[2]E[1])
\}$,
where $\bJ
= e^{\mathrm{ad}(\bvarthetas)} \calJ|_{\Map(T[1]\Sigma_2, T^*[2]E[1])}
= e^{\mathrm{ad}(\bvarthetas)} \mu_* \epsilon \, \ev^* J 
|_{\Map(T[1]\Sigma_2, T^*[2]E[1])}$
%
and the canonical $1$-form is
\begin{eqnarray}
\bvartheta_{T^*[2]E[1]}
&=& \iota_{D} \mu_* \ev^* \vartheta_{T^*[2]E[1]}
\nonumber \\
&=& \int_{T[1]\Sigma_2}
\rd \sigma^2 \rd^2 \theta \,
\left(
- \bp_I \bbd \bx^I
+ \frac{1}{2} k_{AB} \bq^A \bbd \bq^B
\right).
\end{eqnarray}
General forms of BFV currents of degree $0$, $1$ and $2$ are
\begin{eqnarray}
\bJ_{(0)(f)}
&=& \int_{T[1]\Sigma_2} \rd \sigma^2 \rd^2 \theta \,
 \epsilon_{(2)} f(\bx),
\nonumber \\
\bJ_{(1)(a,u)}
&=& \int_{T[1]\Sigma_2} \rd \sigma^2 \rd^2 \theta \,
 \epsilon_{(1)} (a_{I}(\bx)\bbd \bx^I
+u_{A} (\bx)\bq^A),
\nonumber \\
\bJ_{(2)(G,K,F,B,E)}({\sigma,\theta})
&=& \int_{T[1]\Sigma_2} \rd \sigma^2 \rd^2 \theta \,
\epsilon_{(0)} 
\left(G^I(\bx) \bp_{I}+ K_A(\bx) \bbd \bq^A
    + \frac{1}{2}F_{AB}(\bx) \bq^A \bq^B
\right.
\nonumber \\
&&
\left.
    + \frac{1}{2}B_{IJ}(\bx) \bbd \bx^I \bbd \bx^J
    +E_{AI}(\bx)\bbd \bx^I \bq^A\right).
\end{eqnarray}
The general formula 
\eqref{supercurrentcommutation02}
produce the following supergeometric current algebra:
\beqa
&& \sbv{\bJ_{(0)(f)}(\epsilon)}{\bJ_{(0)(f^\prime)}(\epsilon^{\prime})}_{PB}
=0,
\nom\\
&&
\sbv{\bJ_{(1)(u,a)}(\epsilon)}
{\bJ_{(0)(f^\prime)}(\epsilon^{\prime})}_{PB}
=0,
\nom\\
&& \sbv{\bJ_{(2)(G,K,F,H,E)}(\epsilon)}
{\bJ_{(0)(f^\prime)}(\epsilon^{\prime})}_{PB}
=
-G^{I}\frac{\del \bJ_{(0)(f^\prime)}}{\del \bx^I}(\epsilon \epsilon^{\prime}),
\nom\\
&& \sbv{\bJ_{(1)(u,a)}(\epsilon)}
{\bJ_{(1)(u^\prime,a^\prime)}(\epsilon^{\prime})}_{PB}
= - \int_{T[1]\Sigma_2}
\rd \sigma^2 \rd^2 \theta \, 
\epsilon_{(1)} \epsilon^{\prime}_{(1)} (k^{AB}u_A u_B^\prime),
\nom\\
&&
\sbv{\bJ_{(2)(G,K,F,B,E)}(\epsilon)}
{\bJ_{(1)(u^\prime,a^\prime)}(\epsilon^{\prime})}_{PB}
\nonumber \\
&& = -\bJ_{(1)(\bar u,\bar \alpha)}(\epsilon \epsilon^{\prime})
- \int_{T[1]\Sigma_2} \rd \sigma^2 \rd^2 \theta \,
(\bbd \epsilon_{(0)}) \epsilon^{\prime}_{(1)}
(G^{I}\alpha^\prime_I - k^{AB}K_A u_B^\prime),
\nom\\
&&
\sbv{\bJ_{(2)(G,K,F,B,E)}(\epsilon)}
{\bJ_{(2)(G^\prime,K^\prime,F^\prime,B^\prime,E^\prime)} 
(\epsilon^{\prime})}_{PB}
    =-\bJ_{(2)(\bar G, \bar K, \bar F,\bar B,\bar E)}
(\epsilon \epsilon^{\prime})
\nom\\&&
\qquad
- \int_{T[1]\Sigma_2} \rd \sigma^2 \rd^2 \theta \,
(\bbd \epsilon_{(0)}) \epsilon^{\prime}_{(0)}
\left[(G^JB^\prime_{JI}+G^{\prime J}B_{JI}
             +k^{AB}(K_A E_{BI}^\prime+E_{AI}K_{B}^\prime))\bbd \bx^I
\right.
\nom\\&&
\qquad
\left.
+(G^IE^\prime_{AI}+G^{\prime I}E_{AI}
                +k^{BC}(K_B F_{AC}^\prime+F_{AC}K_{B}^\prime))\bq^A \right].
\eeqa
Here
\beqa
    &&\bar \alpha=({\it i}_{G} d+d{\it i}_{G})\alpha^\prime+\langle E-dK,u^\prime\rangle,\quad
    \bar u={\it i}_{G} du^\prime +\langle F, u^\prime\rangle, \nom
\\
    &&\bar G=[G,G^\prime], \nom\\
    &&\bar K={\it i}_{G}d K^\prime-{\it i}_{G^\prime}d K+{\it i}_{G^\prime} E+\langle F, K^\prime\rangle, \nom\\
    &&\bar F={\it i}_{G}d F^\prime-{\it i}_{G^\prime}d F+\langle F, F^\prime\rangle, \nom\\
    &&\bar B=(d{\it i}_{G}+{\it i}_{G}d)B^\prime-{\it i}_{G^\prime}d B+\langle E, E^\prime\rangle
    +\langle K^\prime, dE\rangle-\langle dK,E^\prime\rangle+{\it i}_{G^\prime}{\it i}_{G} H, \nom\\
    &&\bar E=(d{\it i}_{G}+{\it i}_{G}d)E^\prime-{\it i}_{G^\prime}d E+\langle E, F^\prime\rangle
    -\langle E^\prime, F\rangle+\langle dF, K^\prime\rangle-\langle dK, F^\prime\label{3bE}
\rangle,
\label{3dalgebra}
\eeqa
where all the terms are evaluated by $\sigma^{\prime}$.
The bracket $[ -,- ]$ is a Lie bracket on $TM$, $i_G$ is an interior product
with respect to a vector field $G$ and
$\langle -,- \rangle$ is a bilinear form on the fiber of $E$
with respect to the metric $k^{AB}$.

The corresponding classical current algebra is the ghost number
zero components of superfields:
\beqa
&& \sbv{\bJ_{(0)(f)}{}_{cl}(\epsilon)}
{\bJ_{(0)(f^\prime)}{}_{cl}(\epsilon^{\prime})}_{PB}
=0,
\nom\\
&&
\sbv{\bJ_{(1)(u,a)}{}_{cl}(\epsilon)}
{\bJ_{(0)(f^\prime)}{}_{cl}(\epsilon^{\prime})}_{PB}
=0,
\nom\\
&& \sbv{\bJ_{(2)(G,K,F,H,E)}{}_{cl}(\epsilon)}
{\bJ_{(0)(f^\prime)}{}_{cl}(\epsilon^{\prime})}_{PB}
=
-G^{I}\frac{\del \bJ_{(0)(f^\prime)}{}_{cl}}{\del x^I}
(\epsilon \epsilon^{\prime}),
\nom\\
&& \sbv{\bJ_{(1)(u,a)}{}_{cl}(\epsilon)}
{\bJ_{(1)(u^\prime,a^\prime)}{}_{cl}(\epsilon^{\prime})}_{PB}
= - \int_{\Sigma_2}
\cepsilon_{(1)}^{(1)} \wedge \cepsilon^{(1)\prime}_{(1)} k^{AB}
u_A u_B^\prime{}_{cl},
\nom\\
&&
\sbv{\bJ_{(2)(G,K,F,B,E)}{}_{cl}(\epsilon)}
{\bJ_{(1)(u^\prime,a^\prime)}{}_{cl}(\epsilon^{\prime})}_{PB}
\nonumber \\
&& = -\bJ_{(1)(\bar u,\bar \alpha)}{}_{cl}(\epsilon \epsilon^{\prime})
- \int_{\Sigma_2} \rd \cepsilon_{(0)} \wedge \cepsilon^{(1)\prime}_{(1)}
(G^{I}\alpha^\prime_I - k^{AB}K_A u_B^\prime){}_{cl},
\nom\\
&&
\sbv{\bJ_{(2)(G,K,F,B,E)}{}_{cl}(\epsilon)}
{\bJ_{(2)(G^\prime,K^\prime,F^\prime,B^\prime,E^\prime)}{}_{cl}
(\epsilon^{\prime})}_{PB}
    =-\bJ_{(2)(\bar G, \bar K, \bar F,\bar B,\bar E)}{}_{cl}
(\epsilon \epsilon^{\prime})
\nom\\&&
\qquad
- \int_{\Sigma_2}
\rd \cepsilon_{(0)} \cepsilon^{\prime}_{(0)}
\left[(G^JB^\prime_{JI}+G^{\prime J}B_{JI}
             +k^{AB}(K_A E_{BI}^\prime+E_{AI}K_{B}^\prime)){}_{cl}
\rd x^I
\right.
\nom\\&&
\qquad
\left.
+(G^IE^\prime_{AI}+G^{\prime I}E_{AI}
                +k^{BC}(K_B F_{AC}^\prime+F_{AC}K_{B}^\prime)){}_{cl}
q^A \right],
\eeqa
where $\cepsilon_{(i)}^{(i)} = \frac{1}{i!} dx^{\mu_1} \wedge
\cdots \wedge dx^{\mu_i} \epsilon_{(i) \mu_1 \cdots \mu_i}(\sigma)$
is the $i$-form part of the test function and 
and $q^A = d x^{\mu} q^A_{\mu}$.
This coincides with the generalized current algebra in two dimensions
in \cite{Ikeda:2011ax}.
\end{example}

\begin{example}[Current algebras of topological $n$-branes]
\cite{Bonelli:2005ti}
Consider an $n+1$ dimensional spacetime $X = \Sigma_n \times \bR$, where
$\Sigma_n$ is a (compact, orientable) manifold in $n$ dimensions.
We construct a generalized current algebra on the mapping space 
$\Map(\Sigma_n, T^*M)$, where $M$ is a target manifold.

Take a super phase space $\Map(T[1]\Sigma_n, T^*[n]M)$.

First, we consider the target graded manifold $T^*[n]M$.
A cotangent bundle $T^*[n+1]\calM = T^*[n+1]T^*[n]M$
is a graded symplectic manifold of degree $n+1$ with 
a canonical symplectic form.
Take local coordinates $(x^I, p_I)$
of degree $(0, n)$ on $T^*[n]M$,
and conjugate Darboux coordinates $(\xi_I, \chi^I)$ of degree
$(n+1, 1)$ on $T^*[n+1]$.
The canonical graded symplectic structure on $T^*[n+1]\calM$ is
\beq
\omega_{T^*[n+1]T^*[n]M} = \delta x^I\wedge\delta \xi_I
+\delta p_I\wedge\delta \chi^I.
\label{nbranesympform}
\eeq
We take the following horizontal homological function of degree $n+2$:
\beq
\Theta=  \chi^I \xi_I
+ \frac{1}{(n+2)!}
H_{I_1 I_2 \cdots I_{n+2}}(x) \chi^{I_1} \chi^{I_2} \cdots \chi^{I_{n+2}}.
\eeq
It is homological if $H$ is a closed $n+2$-form.

Next, we consider the mapping space $\Map(T[1]\Sigma_n, T^*[n]M)$.
Take local coordinates on $T[1]\Sigma_n$, 
$(\sigma^{\mu}, \theta^{\mu})$ of degree $(0,1)$.
Local coordinates on the mapping space are denoted by boldfaces of 
corresponding local coordinates on the target space.
$\bx^I(\sigma, \theta)$ is a smooth map from $T[1]\Sigma_n$ to $M$
and
$\bp_I(\sigma, \theta) \in \Gamma(T^*[1]\Sigma_n
\otimes \bx^*(T^*[n]M))$
is a superfield of degree $n$.
A graded symplectic form induced from Eq.~\eqref{nbranesympform}
is of degree $1$ and defined as
\beq
\bomega_{T^*[n+1]T^*[n]M}
=\int_{T[1] \Sigma_n}
\rd^n \sigma \rd^n \theta \,
(\delta \bx^I\wedge\delta\vxi_I
+\delta \bp_I\wedge\delta \vchi^I),
\eeq
where
$\vxi_I(\sigma,\theta) \in \Gamma(T^*[1]\Sigma_n \otimes \bx^*(T^*[n+1]M))$
 and
$\vchi^I(\sigma,\theta)
\in \Gamma(T^*[1]\Sigma_2 \otimes \bx^*(T^*[n+1]T^*[n]M))$
are canonical conjugate superfields.
A Q-structure homological function is
\begin{eqnarray}
\bTheta
&=&
\mu_* \ev^* \Theta
\nonumber \\
&=& \int_{T[1]\Sigma_n}
\rd^n \sigma \rd^n \theta \,
\left(
\vchi^I\vxi_I
+ \frac{1}{(n+2)!}
H_{I_1 I_2 \cdots I_{n+2}}(\bx)
\vchi^{I_1} \vchi^{I_2} \cdots \vchi^{I_{n+2}}
\right).
\end{eqnarray}

Next we consider functions for currents on the target space. 
The space is the set of functions
of degree equal to or less than $n$,
$\Cn(T^*[n+1]T^*[n]M)
= \{f \in C^{\infty}(T^*[n+1]T^*[n]M) | |f| \leq n \}$.
Here we only consider functions of degree $n$
on $\Cn(T^*[n+1]T^*[n]M)$ because the currents
constructed from functions of degree less than $n$
have trivial Poisson brackets.
A general form of a function of degree $n$ is 
\begin{eqnarray}
J_{(n)(G,B)} &=&
    G^I(x) p_{I}
+\frac{1}{n!}B_{I_1\cdots I_{n}}(x)\chi^I \cdots \chi^{I_{n}},
\end{eqnarray}
where $G^I(x)$ and $B_{I_1\cdots I_{n}}(x)$ are local functions of $x$.
$G^I(x) \partial_I$ is a vector field on $M$ and
$B= \frac{1}{n!} B_{I_1\cdots I_{n}}(x) dx^{I_1} \cdots dx^{I_n}$ 
is an $n$-form on $M$.
Take the space of twisted functionals on the mapping space
$\calC\calA_{\leq n}(T^*[n+1]T^*[n]M))
 = \{\calJ \in
C^{\infty}(\Map(T[1]\Sigma_n, T^*[n+1]T^*[n]M))
| \bJ = e^{\mathrm{ad}(\bvartheta_{T^*[n]M})}
\mu_* \epsilon \, \ev^* J, J \in \Cn(T^*[n+1]T^*[n]M), |J|=n
\}$.
%
A twisting function is given by 
\begin{eqnarray}
\bvarthetas
&=& \iota_{\hat{D}} \mu_* \ev^* \vartheta_s
= \int_{T[1]\Sigma_n}
\rd^n \sigma \rd^n \theta \,
(-1)^{n+1}
\bp_I \bbd \bx^I
.
\end{eqnarray}
%
Then the BFV current $J_{(n)}$ is
\begin{eqnarray}
\bJ_{(n)(G,B)}({\sigma,\theta})
&=& \int_{T[1]\Sigma_n} \rd^n \sigma \rd^n \theta \,
\epsilon_{(0)}
\left(G^I(\bx) \bp_{I}
+\frac{1}{n!}B_{I_1\cdots I_{n}}(\bx)
\bbd \bx^I \cdots \bbd \bx^{I_{n}}
\right),
\end{eqnarray}
with a test function $\epsilon_{(0)}$.
Applying the formula \eqref{supercurrentcommutation02},
we obtain the current algebra as follows:
\beqa
&&
\sbv{\bJ_{(n)(G,B)}(\epsilon)}
{\bJ_{(n)(G^\prime,B^\prime)}
(\epsilon^{\prime})}_{PB}
= -\bJ_{(n)(\courant{J_1}{J_2})}
(\epsilon \epsilon^{\prime})
\nom\\&&
\qquad
-
e^{\mathrm{ad}(\bvartheta_{T^*[n]M})}
\int_{T[1]\Sigma_n} \rd^n \sigma \rd^n \theta \,
 (\bbd \epsilon_{(0)}) \epsilon^{\prime}_{(0)}
\ev^* \bracket{J_{(n)(G,B)}}{J_{(n)(G^\prime,B^\prime)}}.
\eeqa
Here
$\courant{J_1}{J_2}$ is the higher Dorfman bracket
on $TM \oplus \wedge^n T^*M$
 defined by
\begin{eqnarray}
\courant{(u,a)}{(v,b)}
= [u, v] + L_u b - \iota_v d a,
\end{eqnarray}
for $u, v \in \Gamma(TM)$ and $a, b \in \Gamma(\wedge^n T^*M)$,
and $\bracket{J_1}{J_2}
= \iota_u b + \iota_v a$ is a pairing,
$\bracket{-}{-}:(TM \oplus \wedge^n T^*M) \times
(TM \oplus \wedge^n T^*M) \rightarrow \wedge^{n-1} T^*M$.
The classical part of the equations,
\beqa
&&
\sbv{\bJ_{(n)(G,B)}{}_{cl}(\epsilon)}
{\bJ_{(n)(G^\prime,B^\prime)}{}_{cl}
(\epsilon^{\prime})}_{PB}
= -\bJ_{(n)(\courant{J_1}{J_2})}{}_{cl}
(\epsilon \epsilon^{\prime})
\nom\\&&
\qquad
-
\frac{1}{(n-1)!}
\int_{\Sigma_n} \rd \cepsilon_{(0)} \cepsilon^{\prime}_{(0)}
\ev^* \bracket{J_{(n)(G,B)}}{J_{(n)(G^\prime,B^\prime)}}
|_{\chi^I = \rd x^I}
\eeqa
coincides with a generalized current algebra of 
a topological $n$-brane theory 
in \cite{Bonelli:2005ti}.
\end{example}

\begin{example}[Homotopy algebroid current algebras
in higher dimensions]\cite{Ikeda:2011ax}
For general $n$, more general current algebras connected QP-manifold of degree $n$ on $n$-dimensional spaces $\Sigma_n$ \cite{Severa:2001}
(These structures are called Lie $n$-algebroid structures.)
is obtained by our formula \ref{supercurrentcommutation02}.
\end{example}

\section{Current algebras in AKSZ sigma models}\label{CAofAKSZ}
\noindent
Current Algebras in AKSZ sigma models are constructed
by the supergeometric formalism of current algebras in
the previous section.

If we take the commutative subspace $\Bn(T^*[n+1]\calM, \alpha)$
of the space $\Cn(T^*[n+1]\calM)$ twisted by the function $\alpha$
in Eq.~\eqref{Bntwist},
a current algebra without anomalies is also given
by the formula in Theorem \ref{supercurrentalgebratheorem}.
These current algebras are ones in AKSZ sigma models.

Take a QP manifold of degree $n+1$, $(T^*[n+1]\calM, \omega, \Theta, \alpha)$,
where $\Theta$ is assumed a horizontal homological function.
and $\alpha$ is a twisting function of degree $n$.
Suppose that the restricted derive bracket
$\ssbv{-}{-} = \bsbv{\bsbv{-}{\Theta}}{-}|_{\calM}$ is nondegenerate.
Then $\ssbv{-}{-}$ is a graded Poisson bracket of degree $n$,
and $(\calM, \omega_{\calM})$ is graded symplectic.

A QP structure on $\Map(T[1]\Sigma_n, T^*[n+1]\calM)$
is induced by the AKSZ construction,
where the QP structures are given by $\bomega_{T^*[n+1]\calM} 
= \mu_* \ev^* \omegab$ and $\bTheta
=
\mu_* \ev^* \Theta$.

At first, we consider the simplest case, $\alpha=0$,
$\Bn(T^*[n+1]\calM, 0)$.
By a transgression map, we obtain a corresponding space on the mapping space,
$\calB\calA_{\leq n}(T^*[n+1]\calM, 0)
= \{\calJ = \mu_* \epsilon \, \ev^* J \in
C^{\infty}(\Map(T[1]\Sigma_n, T^*[n+1]\calM))
| J \in \Bn(T^*[n+1]\calM, 0), J|_{T^*[n+1]} =0
\}$.

Using our formula \ref{supercurrentalgebratheorem}, we obtain
a current algebra on this space.
This current algebra is trivial 
$\sbv{\calJ_1}{\calJ_2}_{PB} = \ssbv{\calJ_1}{\calJ_2} =0$
for any $\calJ_1, \calJ_2 \in \calB\calA_{\leq n}(T^*[n+1]\calM, 0)$
since $J_1$ and $J_2$ commute.

Next we consider the space twisted by $\alpha$,
$\Bn(T^*[n+1]\calM, \alpha)
= \{e^{\mathrm{ad} (\alpha)} f | f \in \Bn(T^*[n+1]\calM,0) \}$.
Since $\Theta$ is a horizontal homological function,
$\Bn(T^*[n+1]\calM, \alpha)$
is a Poisson algebra by the projected derived bracket.
Twist by $\bvarthetas$
gives a nontrivial current algebra 
if we choose proper nonzero $\alpha$.

Two twists are considered above, 
a twist $\alpha$ and a twist by $\bvarthetas$. They are
unified to the twist by the following one functional on the mapping space
 as
\begin{eqnarray}
\tiS
&=& \bvarthetas + \balpha
= \iota_{D} \mu_* \ev^* \varthetas
+ \mu_* \ev^* \alpha.
\label{AKSZtwistedpullback}
\end{eqnarray}
The space of BFV currents is $\Bn(T^*[n+1]\calM, 0)$
twisted by Equation (\ref{AKSZtwistedpullback}),
\begin{align}
&\calB\calA_{\leq n}(T^*[n+1]\calM, \alpha)
\nonumber \\
&: = \{\bJ = e^{\mathrm{ad}(\tiS)} \mu_* \epsilon \, \ev^* J \in
C^{\infty}(\Map(T[1]\Sigma_n, T^*[n+1]\calM))
| J \in \Bn(T^*[n+1]\calM, 0)
\}.
\end{align}
Since $\ssbv{\bvarthetas}{\balpha}=0$,
this space is equivalent to the space of functionals 
induced from $\Bn(T^*[n+1]\calM, \alpha)$
twisted by $\bvarthetas$:
\begin{eqnarray}
\{\bJ = e^{\mathrm{ad}(\bvarthetas)} \mu_* \epsilon \, \ev^* J \in
C^{\infty}(\Map(T[1]\Sigma_n, T^*[n+1]\calM))
| J \in \Bn(T^*[n+1]\calM, \alpha)
\}.
\end{eqnarray}
The space $\calB\calA_{\leq n}(T^*[n+1]\calM, \alpha)$ is
a Poisson algebra, i.e. a current algebra on $\Map(T[1]\Sigma_n, \calM)$
by the Poisson bracket $\ssbv{-}{-} = \sbv{-}{-}_{PB}$ of degree zero.
We can refer to a homological function of an 
AKSZ sigma model,
or a twisted AKSZ sigma model
by $\tiS$ \cite{Ikeda:2012pv, Ikeda:2013wh}.
In fact, the Poisson algebra 
$\calB\calA_{\leq n}(T^*[n+1]\calM, \alpha)$
is equivalent to
the BFV current algebra of a (twisted) AKSZ sigma model
with the Q-structure action (\ref{AKSZtwistedpullback})
as in examples belows.

If $\ssbv{\tiS}{\tiS}=0$, which is a special case,
$\tiS$ defines the (nontwisted) genuine AKSZ sigma model
on
$\Map(T^*[1]\Sigma_n, \calM)$
with the Q-structure $\tiS$,
which is equivalent to the BFV formalism of the AKSZ sigma model.


\begin{example}[$n=1$: Poisson Sigma Models]\label{examplePSM}
\noindent
Let us take the same setting as the twisted Poisson structure in Example
\ref{examplepoisson} with $H=0$,
a QP-structure of degree $2$ on $T^*[2]\calM = T^*[2]T^*[1]M$
with the canonical symplectic form $\omega_{T^*[2]T^*[1]M}$,
the horizontal homological function $\Theta = \xi_I q^I$,
and a function of degree $2$, $\alpha = \frac{1}{2} \pi^{IJ}(x) p_I p_J$.

Let $X = S^1 \times \bR$ be a two-dimensional worldsheet,
and take a supermanifold $T[1]S^1$.
The AKSZ construction defines a QP structure
of degree $1$ on $\Map(T[1]S^1, T^*[2]T^*[1]M)$.
The symplectic structure of degree $1$
on $\Map(T[1]S^1, T^*[2]T^*[1]M)$ is
\begin{eqnarray}
\bomega_{T^*[2]T^*[1]M}
&=& \mu_* \ev^* \omega_{T^*[2]T^*[1]M}
= \int_{T[1]S^1}
\rd \sigma \rd \theta \,
\left(\delta \bx^{I} \wedge \delta \bxi_{I}
+ \delta \bp_{I} \wedge \delta \bq^{I}
\right).
\end{eqnarray}
A horizontal homological function is
\begin{eqnarray}
\bTheta
&=&
\mu_* \ev^* \Theta
= \int_{T[1]S^1}
\rd \sigma \rd \theta \,
\bxi_I \bq^I
.
\end{eqnarray}
The projected derived bracket
$\ssbv{-}{-}
= \bsbv{\bsbv{-}{\bTheta}}{-}|_{T^*[1]M}
$ is nondegenerate. In fact, the symplectic form 
on $\Map(T[1]S^1, T^*[1]M)$,
is $\bomega_{T^*[1]M} = \int_{T[1]S^1}
\rd \sigma \rd \theta \, \delta \bx^i \wedge \delta \bp_I
$.

The twisting function in Equation (\ref{AKSZtwistedpullback})
on the mapping space is
\begin{eqnarray}
S_{T^*[1]M}
&=& \int_{T[1]S^1}
\rd \sigma \rd \theta \,
\
\left(
\bp_I \bbd \bx^I
+ \frac{1}{2} \pi^{IJ}(\bx) \bp_I \bp_J
\right),
\end{eqnarray}
which is eq uivalent to
the BRST charge of the BFV formalism of the Poisson sigma model
\cite{Ikeda:1993aj, Schaller:1994es} if we expand superfields to component fields.

In order to construct the current algebra of this model,
we take the basis $(x^I, p_I)$ of $B_{\leq 1}(T^*[2]T^*[1]M,0)$
since $B_{\leq 1}(T^*[2]T^*[1]M,0) \simeq C^{\infty}(T^*[1]M)$.
We twist the basis,
\begin{eqnarray}
J^I_{(0)}
&=& e^{\mathrm{ad} (\alpha)} x^I
= 0,
\nonumber \\
J^I_{(1)}
&=& e^{\mathrm{ad} (\alpha)} q^I
= q^I + \pi^{IJ}(x) p_J,
\end{eqnarray}
Thus we obtain supergeometric currents,
\begin{eqnarray}
&& \bJ^I_{(0)} = 0,
\\
&& \bJ^I_{(1)}
= e^{\mathrm{ad}(S_{T^*[1]M})} \mu_* \epsilon \, \bq^I
= \int_{T[1]S^1} \rd \sigma \rd \theta \, 
\epsilon_{(0)}
\, (\bbd \bbx^{I} + \pi^{IJ}(\bbx) \bp_{J}).
\end{eqnarray}
on
$\calB\calA_{\leq 1}(T^*[2]T^*[1]M, \alpha)$.

The Poisson bracket of the current $\bJ^I_{(1)}$
is calculated 
using the identity (\ref{poissonjacobi}) with $H=0$
and the formula (\ref{supercurrentcommutation}):
\begin{eqnarray}
\sbv{\bJ_{(1)}^{I}(\epsilon)}
{\bJ_{(1)}^{J}(\epsilon^{\prime})}_{PB} &=&
- e^{\mathrm{ad}(S_{T^*[1]M})} \mu_* \epsilon_{(0)}  \epsilon^{\prime}_{(0)} 
\ev^* \bsbv{\bsbv{J^I}{\Theta}}{J^J}|_{\Map(T[1]S^1, T^*[1]M)} \,
\nonumber \\
&=&
- e^{\mathrm{ad}(S_{T^*[1]M})} \mu_* \epsilon_{(0)}  \epsilon^{\prime}_{(0)}  \,
\ev^* \frac{\partial \pi^{IJ}}{\partial x^{K}}(x) J^K.
\end{eqnarray}
Note that there is no anomalous term.
The classical part of this equation becomes the correct current algebra of
the Poisson sigma model.
Taking the ghost number zero part, we obtain the physical current algebra,
\begin{eqnarray}
\sbv{\cJ_{(1)}^{I}(\sigma)}
{\cJ_{(1)}^{J}(\sigma^{\prime})}_{PB} &=&
- \frac{\partial \pi^{IJ}}{\partial x^{K}}(x(\sigma))
\cJ^{K}_{(1)}(\sigma)
\delta(\sigma-\sigma^{\prime}),
\end{eqnarray}
where $\cJ_{(1)}^{I}(\sigma) = \partial_{\sigma} x^{I}(\sigma)
+ \pi^{IJ}(x(\sigma)) p_{J}(\sigma).$

\end{example}
\vskip10pt

\begin{example}[$n=2$: Courant sigma models]
Consider a QP-structure of degree $3$ on 
$T^*[3]\calM = T^*[3]T^*[2]E[1]$,
where $E$ is a vector bundle on a smooth manifold $M$ and $\calM = T^*[2]E[1]$.
Let $X = \Sigma_2 \times \bR$ be a manifold of dimension $3$
where $\Sigma_2$ is a compact (oriented) manifold of dimension $2$.
We take a base supermanifold $T[1]\Sigma_2$.

The AKSZ construction gives a QP structure
of degree $1$ on $\Map(T[1]\Sigma_2, T^*[3]\calM)$.
Take local coordinates on $T^*[3]\calM$ as
$(x^I, q^A, p_I)$ of $(0,1,2)$
and conjugate coordinates $(\xi_I, \eta_A, \chi^I)$ of degree
$(3,2,1)$.
Corresponding local coordinate superfields on the mapping space
are denoted by boldface letters.
$\bbx^{I}$ is a map from $T[1]\Sigma_2$ to $M$,
$\bq^{A}$ is a degree $1$
section of $T^*[1]\Sigma_2 \otimes \bbx^*(E[1])$,
$\bp_{I}$ is a degree $2$
section of $T^*[1]\Sigma_2 \otimes \bbx^*(T^*[2]M)$,
$\bxi_{I}$ is a degree $3$
section of $T^*[1]\Sigma_2 \otimes \bbx^*(T^*[3]M)$,
$\veta_{A}$ is a degree $2$
section of $T^*[1]\Sigma_2 \otimes \bbx^*(T^*[3]E[1])$,
and $\vchi^{I}$ is a degree $1$
section of $T^*[1]\Sigma_2 \otimes \bbx^*(T^*[3]T^*[2]M)$.

The symplectic structure
on $\Map(T[1]\Sigma_2, T^*[3]\calM)$ is
\begin{eqnarray}
\bomega_{T^*[3]\calM}
&=& 
\int_{T[1]\Sigma_2}
\rd^2 \sigma \rd^2 \theta \,
\left(
\delta \bx^I\wedge\delta\vxi_I
    -\delta\bq^A\wedge\delta \veta_A
+\delta \bp_I\wedge\delta \vchi^I
\right).
\end{eqnarray}
A horizontal homological function is chosen as
$\Theta = \chi^I\xi_I +\frac{1}{2}k^{AB}\eta_A\eta_B$,
where $k_{AB}$ is a fiber metric on $E$.
Then the homological function on the mapping space is
\begin{eqnarray}
S_{T^*[3]\calM}
&=& \int_{T[1]\Sigma}
\rd^2 \sigma \rd^2 \theta \,
\left(
\vchi^I\vxi_I
+\frac{1}{2}k^{AB}\veta_A\veta_B
\right),
\end{eqnarray}
which trivially satisfies $\bsbv{S_{T^*[3]\calM}}{S_{T^*[3]\calM}}=0$
and horizontal.
The projected derived bracket
$\ssbv{-}{-}
= \bsbv{\bsbv{-}{\bTheta}}{-}|_{\calM}$ defines
the graded symplectic form $\bomega_{\calM}$ on $\Map(T[1]\Sigma_2, \calM)$
since it is nondegenerate.

We take a twisting function of degree $3$ 
as
$$\alpha =
\rho^I{}_{A}(x) q^A p_I
+\frac{1}{3!} H_{ABC}(x) q^A q^B q^C.$$
Here $\rho$ and $H$ define the bundle map
$\rho = \rho{}^I{}_{A}(x) e^A \frac{\partial}{\partial x^I}
: E \rightarrow TM$
and a section of $\wedge^3 E^*$,
$H= \frac{1}{3!} H_{ABC}(x) e^A \wedge e^B \wedge e^C$,
where $e^A$ is a section on $E$.
We impose $\ssbv{\alpha}{\alpha}=0$, which is equivalent to 
the Courant algebroid structure on $E$.
The twisting function induced from $\varthetas$ and $\alpha$ is
\begin{align}
\tiS
&= \bvarthetas + \balpha
\nonumber \\ 
&= \int_{T[1]\Sigma_2}
\rd^2 \sigma \rd^2 \theta \,
\left(
- \bp_I \bbd \bx^I
+ \frac{1}{2} k_{AB} \bq^A \bbd \bq^B
+
\rho{}^I{}_{A}(\bx) \bq^A \bp_I
+\frac{1}{3!} H_{ABC}(\bx) \bq^A \bq^B \bq^C
\right),
\end{align}
which is of degree $1$ and coincides with
the BRST charge of the BFV formalism of the Courant sigma model
\cite{Ikeda:2002wh, Roytenberg:2006qz}.

In order to construct the current algebra of the Courant sigma model,
we take basis $(x^I, \chi^I, \eta_A)$
of the commutative subspace $B_{\leq 2}(T^*[3]T^*[2]E[1], 0)$,
and twist by $\alpha$.
Basis are
\begin{eqnarray}
J^I_{(0)}
&=&
e^{\mathrm{ad} (\alpha)} x^I
= 0,
\nonumber \\
J^I_{(1)}
&=&
e^{\mathrm{ad} (\alpha)} \chi^I
= \chi^I -
\rho{}^I{}_{A}(x) q^A,
\nonumber \\
J_{(2)A}
&=&
e^{\mathrm{ad} (\alpha)} \eta_A
= \eta_A -
\rho{}^I{}_{A}(x) p_I
- \frac{1}{2} H{}_{ABC}(x) q^B q^C.
\end{eqnarray}
Thus, supergeometric BFV currents are obtained as
\begin{eqnarray}
\bJ^I_{(0)} &=& 0,
\\
\bJ^I_{(1)}
&=&
e^{\mathrm{ad}(S_{\calM})} \mu_* \epsilon \, \ev^* \chi^I
=
\int_{T[1]\Sigma_2} \rd^2 \sigma \rd^2 \theta \, \epsilon_{(1)} \,
(\bbd \bbx^{I}
- \rho{}^I{}_{A}(\bx) \bq^A),
\nonumber \\
\bJ^I_{(2)}
&=&
e^{\mathrm{ad}(S_{\calM})} \mu_* \epsilon \, \ev^* \eta_A
\nonumber \\
&=&
\int_{T[1]\Sigma_2} \rd^2 \sigma \rd^2 \theta \, \epsilon_{(0)} \,
(k_{AB} \bbd \bq^B -
\rho{}^I{}_{A}(\bx) \bp_I
- \frac{1}{2} H{}_{ABC}(\bx) \bq^B \bq^C).
\end{eqnarray}


The current algebra is calculated by the formula 
\eqref{supercurrentcommutation02}.
It is
\beqa
  \{\bJ_{(1)}^I(\epsilon),
\bJ_{1}^J(\epsilon^\prime)\}_{PB}&=&0,
\nonumber \\
  \{\bJ_{(2)A}(\epsilon),
\bJ_{(1)}^I(\epsilon^\prime)\}_{PB}
     &=&
- e^{\mathrm{ad}(S_{\calM})} \mu_* \epsilon_{(0)} \, \epsilon^{\prime}_{(1)} \, \ev^*
\frac{\partial \rho^I{}_{A}}{\partial x^J} (x)
J_{(1)}^J,\nom\\
  \noalign{\vskip 2mm}
   \{\bJ_{(2)A}(\epsilon),
\bJ_{(2)B}(\epsilon^\prime)\}_{PB}
     &=&
e^{\mathrm{ad}(S_{\calM})} \mu_* \epsilon_{(0)} \, \epsilon^{\prime}_{(0)} \, \ev^*
\left(\frac{\partial H_{ABC}}{\partial x^I}(x) J_{(1)}^I q^{C}
+ H_{ABC}(x) k^{CD} J_{(2)D}\right).
\nonumber
\eeqa
Taking the classical part of this equation, we obtain the
correct current algebra of the Courant sigma model:
\beqa
\{\cJ_{(1)i}^I({\sigma}), \cJ_{(1)j}^J({\sigma}^\prime)\}_{PB}&=&0,
\nonumber \\
\{\cJ_{(2)A}({\sigma}), \cJ_{(1)i}^I({\sigma^\prime})\}_{PB}
&=& - \frac{\partial \rho^I{}_{A}}{\partial x^J} (x)
\cJ_{(1)i}^J ({\sigma})
\delta^2({\sigma}-{\sigma}^\prime),
\nom\\
\{\cJ_{(2)A}({\sigma}), \cJ_{(2)B}({\sigma}^\prime)\}_{PB}
&=& {\Big (}\frac{\partial H_{ABC}}{\partial x^I}(x)
(\cJ_{(1)1}^I ({\sigma})q_2^{C}({\sigma})
- \cJ_{(1)2}^I ({\sigma})q_1^{C}({\sigma}))
\nom\\
&& + H_{ABC}(x)k^{CD}\cJ_{(2)D}({\sigma}){\Big )}
\delta^2({\sigma}-{\sigma}^\prime),
\nonumber
\eeqa
where
\begin{eqnarray}
   \cJ_{(1)i}^I(\sigma)&=&\partial_{i} x^I(\sigma)
+ \rho^I{}_{A}(x(\sigma)) q_{i}^A(\sigma),
\nom\\
   \cJ_{(2)12, A}(\sigma)
&=&k_{AB}(\partial_1q^B_2({\sigma})-\partial_2q^B_1({\sigma}))
   +\rho^I{}_{A}(x({\sigma}))p_{I12}({\sigma})
   +H_{ABC}(x({\sigma}))q_1^B({\sigma})q_2^C({\sigma}).
\label{currentCSM}
\nonumber
\end{eqnarray}

\end{example}

\section{Summary and outlook}
\noindent
We have reformulated the current algebra theory by supergeometry 
and constructed new current algebras with
higher structures induced from QP structures.
This new formulation makes not only the calculations easier
but also shed light on new mathematical structures of the current algebras.
We reformulated current algebra based on algebroids and obtained new current 
algebras based on higher algebroids.

Extension of our theory to more general
current algebras including higher derivative terms
\cite{Ekstrand:2009qz} is interesting.
Another recent analysis for a current algebra theory appears in
\cite{Alekseev:2010gr, ALEKSEEV:2013ypa}.
It is interesting to analyze relations to our theory.

Though we have analyzed classical aspect of the current algebra,
a quantum version of the theory is important for mathematical properties 
and applications of current algebras.

The formulation in Theorem \ref{supercurrentalgebratheorem}
is based on an explicit expression of a set of functions 
under the derived bracket including a restriction.
A general theory for the derived bracket and restriction
is remained to be understood.

\subsection*{Acknowledgments}
\noindent
We would like to thank Pavol \v{S}evera for useful comments on the preliminary version of the manuscript. X.X.~is especially grateful to Anton Alekseev for inspiring conversations and helpful suggestions in writing this paper.
N.I.~is supported by the grant of Maskawa Institute and by JSPS Grants-in-Aid for Scientific Research Number 22K03323.
X.X.~is supported by the grant PDFMP$2_{-}$141756 
of the Swiss National Science Foundation.

\subsection*{Data Availability}
\noindent
The data that supports the findings of this study are available within the article.

\appendix
\section{Super functions on twisted Poisson manifolds}
\noindent
Here, we explain a typical example of the space of current functions
in the $n=1$ case.

\begin{example}[$n=1$: Twisted Poisson structure]\label{examplepoisson}

Let $\calM = T^*[1]M$ be a graded manifold of degree 1
on a smooth manifold $M$,
and $(T^*[2]\calM, \omega_{T^*[2]\calM}, \Theta)$ a QP-manifold of degree $2$.
Take local coordinates $(x^I, p_I, q^I, \xi_I)$ of degree $(0,1,1,2)$,
where $(\xi_I, q^I)$ is the conjugate basis of $(x^I, p_I)$, i.e.
$\calM = T^*[1]M = \{(x^I, p_I, q^I, \xi_I)|\xi_I = q^I =0 \}$.
The canonical graded symplectic structure is
$$
\omega_{T^*[2]\calM} = \delta x^I \wedge \delta \xi_I 
+ \delta p_I \wedge \delta q^I,
$$
which defines a Poisson bracket $\bsbv{-}{-}$ such that
$\bsbv{x^I}{\xi_J} = \bsbv{p_J}{q^I} = \delta^I{}_J$,
and other Poisson brackets are zero.

Since a homological function $\Theta$ is of degree $3$, 
we can assume $\Theta$ as
\begin{eqnarray}
\Theta = \xi_I q^I + \frac{1}{3!} H_{IJK}(x) q^I q^J q^K,
\end{eqnarray}
Here $H$ defined by $H = \frac{1}{3!} H_{IJK}(x) dx^I \wedge dx^J \wedge dx^K$ 
is a closed $3$-form on $M$.
Then, $\Theta$ satisfies
$\bsbv{\Theta}{\Theta}=0$,
and this $\Theta$ is a horizontal homological function.

Let us consider the space of current functions for $n=1$,
$$
C_1(T^*[2]T^*[1]M)
= \{f \in C^{\infty}(T^*[2]T^*[1]M)| |f| \leq 1 \}.
$$
Elements are functions of degree $0$ and $1$ on $T^*[2]\calM$.
A function of degree $0$ is a function on $M$, $J_{(0)f} = f(x)$.
A functions of degree $1$ is linear of $q^I$ and $p_I$,
$J_{(1)(u, a)} = a_I(x) q^I + u^I(x) p_I$.
$J_{(1)(u, a)}$ is regarded as a section of $TM \oplus T^*M$
since $C^{\infty}(T[1]M)$ is identified to $\Gamma(\wedge^{\bullet}(T^*M))$,
and $C^{\infty}(T^*[1]M)$ is identified to $\Gamma(\wedge^{\bullet}(TM))$.
That is,
$a_I(x) q^I + u^I(x) p_I$
is identified with
$a_I(x) d x^I + u^I(x) \frac{\partial}{\partial x^I}$.
As explained in the general theory, a small graded Poisson bracket on 
$C_1(T^*[2]T^*[1]M)$
is given by the pullback derived bracket,
$\ssbv{-}{-} = \proj_* \bsbv{\bsbv{ -}{\Theta}}{ -}$.
This small bracket is parametrized by $(x^I, p_I)$ and
satisfies $\ssbv{x^I}{p_J} = \delta^I{}_J$
and the Poisson bracket is nondegenerate.

Let
$J_{(0)(v)}^{\prime} = g(x) \in C_1(T^*[2]T^*[1]M)$ and
$J_{(1)(v, b)}^{\prime} = b_I(x) q^I + v^I(x) p_I
\in C_1(T^*[2]T^*[1]M)$.
The derived brackets of functions are computed as
\begin{eqnarray}
\bsbv{\bsbv{J_{(0)(f)}}{\Theta}}{J_{(0)(g)}^{\prime}} &=& 0,
\label{derivedbrackettwistedPoisson01}\\
\bsbv{\bsbv{J_{(1)(u, a)}}{\Theta}}{J_{(0)(g)}^{\prime}}
 &=& -u^{I}\frac{\del J_{(0)(g)}^{\prime}}{\del x^I},
\label{derivedbrackettwistedPoisson02}\\
\bsbv{\bsbv{J_{(1)(u, a)}}{\Theta}}{J_{(1)(v, b)}^{\prime}}
&=& -
\left[\left(u^J\frac{\del v^{I}}{\del x^J}
    -v^{J}\frac{\del u^I}{\del x^J}\right) p_I\right.
\nom\\
    && \left.+\left(u^J\frac{\del b_I}{\del x^J}
    -v^{J}\frac{\del a_I}{\del x^J}
    +v^{J}\frac{\del a_J}{\del x^I}
    + b_J \frac{\del u^J}{\del x^I}
    +H_{JKI} u^{J} v^{K}
\right) q^I \right]
\nonumber \\
    &=& - J_{(1)(\courant{(u,a)}{(v,b)})}.
\label{derivedbrackettwistedPoisson03}
\end{eqnarray}
where $\courant{(u,a)}{(v,b)}$ is a Dorfman bracket on
$TM \oplus T^*M$ defined by
\begin{eqnarray}
\courant{(u,a)}{(v,b)}
&=& [u, v] + L_u b - \iota_v d a + \iota_u \iota_v H,
\label{2dsupercommutation}
\end{eqnarray}
for $u, v \in \Gamma(TM)$, $a, b \in \Gamma(T^*M)$,
$H$ is a closed $3$-form,
$[u,v]$ is a Lie bracket of vector fields,
$L_u$ is a Lie derivative and
$\iota_v$ is the interior product.
\eqref{derivedbrackettwistedPoisson01}--
\eqref{derivedbrackettwistedPoisson03}
induce graded Poisson brackets on the small P-manifold $\calM$:
\begin{eqnarray}
&& \ssbv{J_{(0)(f)}}{J_{(0)(g)}^{\prime}} =0,
\nom\\
&&
\ssbv{J_{(1)(u,\alpha)}}{J_{(0)(g)}^{\prime}} = -u^{I}
\frac{\del J_{(0)(g)}^{\prime}}{\del x^I},
\nom\\
&&
\ssbv{J_{(1)(u,a)}}{J_{(1)(v,b)}^{\prime}}
    = - J_{(1)([u,v],0)}.
\label{smallsupercommutation}
\end{eqnarray}

Next, we identify $Comm_1(T^*[2]\calM)$, which 
is a Poisson commutative subspace with respect to the big Poisson bracket 
$\bsbv{-}{-}$. Big Poisson brackets of current functions are
\begin{eqnarray}
\bsbv{J_{(0)(f)}}
{J_{(0)(g)}^{\prime}} &=&0,
\nom\\
\bsbv{J_{(1)(u,a)}}
{J_{(0)(g)}^{\prime}}
&=& 0.
\nom\\
\bsbv{J_{(1)(u,a)}}
{J_{(1)(v,b)}^{\prime}}
&=& a_I v^I + u^I b_I = \bracket{(u, a)}{(v, b)}.
\label{innerproductn1}
\end{eqnarray}
Here $\bracket{(u, a)}{(v, b)}$ is the inner product on $TM \oplus T^*M$.
Therefore the commutative subspace $(Comm_1(T^*[2]\calM),\{-,-\}_s)$
is equivalent to the subspace of functions
$J_{(i)}$'s with $\bracket{(u, a)}{(v, b)}=0$,
which is nothing but the Dirac structure \cite{
LWX}.

We deform the above Poisson algebra
by introducing a function $\alpha$ 
of degree $2$. We take $\alpha = \frac{1}{2} \pi^{IJ}(x) p_I p_J$, 
which is a pullback of the same function on $\calM$.

We consider a special condition for $\alpha$ as an example.
Let $\calL$ be a Lagrangian submanifold of a QP-manifold $\calN$.
Then a function $\alpha \in C^{\infty}(\calN)$ is called a 
\textsl{canonical function} respect to $\calL$ 
if $e^{\delta_{\alpha}}\Theta|_{\calL}=
\proj_* e^{\delta_{\alpha}}\Theta=0
$ \cite{Ikeda:2013wh}.
This condition is related to a boundary condition of AKSZ sigma models
with boundary and important to analyze current algebras 
of AKSZ sigma models.
It is also considered as a generalized Maurer-Cartan equation for $\alpha$.
In this article, we take $\calL = \calM$ and $\calN = T^*[n+1]\calM$.
In fact, 
\begin{eqnarray}
\proj_* e^{\delta_{\alpha}} \Theta = 
\proj_*(\Theta +
\sbv{\Theta}{\alpha}_b
+ \frac{1}{2} \sbv{\sbv{\Theta}{\alpha}_b}{\alpha}_b
+ \frac{1}{3!} \sbv{\sbv{\sbv{\Theta}{\alpha}_b}{\alpha}_b}{\alpha}_b
+ \cdots).
\end{eqnarray}
The second term is the derived bracket, which gives
the small bracket $\{\alpha,\alpha\}_s$.

Suppose that 
$\alpha = \frac{1}{2} \pi^{IJ}(x) p_I p_J$ is a canonical function,
$e^{- \delta_{\alpha}} \Theta |_{\calM} =0$.
This condition is equivalent to the following equation:
\begin{eqnarray}
\frac{\partial \pi^{IJ}}{\partial x^L} \pi^{LK}
+ (IJK \ \mbox{cyclic}) = \pi^{IL} \pi^{JM} \pi^{KN} H_{LMN}.
\label{poissonjacobi}
\end{eqnarray}
i.e. the bivector field $\pi = \frac{1}{2} \pi^{ij}(x)
\frac{\partial}{\partial x^i} \wedge \frac{\partial}{\partial x^i}$
satisfies $[\pi, \pi]_S = \wedge^3 \pi^{\sharp} H$,
which a twisted Poisson structure on $M$
\cite{Park2000au, Klimcik:2001vg, Severa:2001qm}.
We obtain different small brackets on currents 
from Equations (\ref{smallsupercommutation}) as follows.

The subspace of commutative functions twisted by $\alpha$ is
$B_1(T^*[2]T^*[1]M, \alpha)
= \{e^{\mathrm{ad} (\alpha)} f| f \in B_1(T^*[2]T^*[1]M, 0)
 \}$.
Since the twisted derived bracket 
$\bsbv{\bsbv{J_1}{e^{- \delta \alpha} \Theta}}{J_2}$
is equivalent to
$\bsbv{\bsbv{e^{\mathrm{ad} (\alpha)} J_1}{\Theta}}{e^{\mathrm{ad} (\alpha)} J_2}$
by a canonical transformation,
we calculate the latter equation.
We take the basis 
$K^I_{(0)} = x^I$ and $K^I_{(1)} = q^I$
for $B_1(T^*[2]T^*[1]M, 0)
$.
Their graded Poisson brackets and derived brackets vanish.
Twisting by $\alpha$,
the basis changes to
\begin{eqnarray}
&&
J^I_{(0)}
= e^{\mathrm{ad} (\alpha)} x^I = x^I,
\nonumber \\
&&
J^I := J^I_{(1)}
= e^{\mathrm{ad} (\alpha)} q^I
= q^I + \pi^{IJ}(x) p_J,
\end{eqnarray}
an obtain the space ${B}_1(T^*[2]T^*[1]M, \alpha)$.
The Poisson bracket of $J^I$'s are zero
$\bsbv{J^I}{J^J}=0$,
but the derived bracket is
\begin{eqnarray}\label{Poissoncommutation}
\bsbv{\bsbv{J^I}{\Theta}}{J^J}
= - \left(
\frac{\partial \pi^{IJ}}{\partial x^K}
+ \pi^{IL} \pi^{JM} H_{LMK}
\right)
J^K,
\end{eqnarray}
using Equation (\ref{poissonjacobi}). 
This derives the small Poisson bracket on the small manifold $\calM$,
\begin{eqnarray}
\ssbv{\proj_* J^I}{\proj_* J^J}
= \proj_* \bsbv{\bsbv{J^I}{\Theta}}{J^J} 
= -\left(
\frac{\partial \pi^{IJ}}{\partial x^K}
+ \pi^{IL} \pi^{JM} H_{LMK}
\right) \proj_* J^K
,
\label{Poissoncommutation}
\end{eqnarray}
where $\proj_* J^K
= \pi^{KL} p_L$.
This gives a Poisson algebra 
associated to twisting by the canonical function $\alpha$.

\end{example}

\newcommand{\bibit}{\sl}



\end{document}